\documentclass[]{article}
\usepackage[a4paper, total={6in, 8in}]{geometry}

\usepackage{amsfonts}
\usepackage{amsmath}
\usepackage{amssymb}

\usepackage{mathtools} 

\usepackage{bm}
\usepackage{cancel}

\usepackage{authblk}

\usepackage{hyperref}

\usepackage{amsthm}
\newtheorem{theorem}{Theorem}
\newtheorem{lemma}[theorem]{Lemma}
\newtheorem{proposition}[theorem]{Proposition}
\newtheorem{corollary}[theorem]{Corollary}
\newtheorem{conjecture}[theorem]{Conjecture}
\newtheorem{definition}[theorem]{Definition}

\theoremstyle{remark}
\newtheorem{remark}[theorem]{Remark}
\newtheorem{example}[theorem]{Example}

\numberwithin{theorem}{section}

\usepackage{tabularx}
\newcolumntype{C}{>{$}c<{$}}
\newcolumntype{L}{>{$}l<{$}}
\newcolumntype{R}{>{$}r<{$}}

\usepackage{multirow}

\usepackage{tikz}
\usetikzlibrary{matrix,decorations.pathmorphing,positioning}

\newcommand*{\h}{\mathfrak{h}}
\newcommand*{\g}{\mathfrak{g}}
\newcommand*{\gs}{\mathfrak{k}}
\newcommand*{\C}{\mathbb{C}}
\newcommand*{\bt}{\tilde{b}}

\newcommand*{\mcC}{\mathcal{C}}
\newcommand*{\mcI}{\mathcal{I}}
\newcommand*{\mf}{\mathfrak}

\newcommand*{\Jmax}{N}

\newcommand*{\X}{X_{\hbar}(\g)}
\newcommand*{\Xs}{X_{\hbar}(\gs)}

\newcommand*{\ot}{\otimes}

\DeclareMathOperator{\End}{End}
\DeclareMathOperator{\Hom}{Hom}
\DeclareMathOperator{\vecop}{vec}
\DeclareMathOperator{\tr}{tr}
\DeclareMathOperator*{\Res}{Res}

\newcommand*{\covec}{\overline{\vecop}}

\newcommand{\draft}[1]{\iffalse #1 \fi}
\usepackage{multicol} 

\title{On the nested algebraic Bethe ansatz for spin chains with simple $\g$-symmetry}
\author{Allan John Gerrard \footnote{a.j.gerrard$\ot$rs.tus.ac.jp}}
\affil{Tokyo University of Science}

\begin{document}
\renewcommand{\theenumi}{(\roman{enumi})}%

\maketitle

\begin{abstract}
	We propose a new framework for the nested algebraic Bethe ansatz for a closed, rational spin chain with $\g$-symmetry for any simple Lie algebra $\g$.
	Starting the nesting process by removing a single simple root from $\g$, we use the residual $U(1)$ charge and the block Gauss decomposition of the $R$-matrix to derive many standard results in the Bethe ansatz, such as the nesting of Yangian algebras, and the AB commutation relation. 
\end{abstract}

\tableofcontents

\section{Introduction}

The Bethe ansatz \cite{betheZurTheorieMetalle1931} in its various forms is one of the central pillars of quantum integrability, facilitating exact calculation of the spectrum for an integrable model.
Ever since the introduction of the algebraic Bethe ansatz (quantum inverse scattering method) \cite{sklyaninQuantumInverseProblem1979} for the $\mathfrak{su}_2$-symmetric Heisenberg spin chain, the extension to higher Lie types has been of interest, prompting the development of the nested algebraic Bethe ansatz \cite{kulishDiagonalisationGLInvariant1983}.
In the nested Bethe ansatz, the steps of the algebraic Bethe ansatz lead to the reduction of the problem to one with a smaller symmetry group -- a nested problem. 
Reducing inductively via $\g = \g_1 \supset \dots \supset \g_r \cong \mathfrak{sl}_2$, one obtains a problem that can be solved by the standard algebraic Bethe ansatz. 
At this stage, the nested algebraic Bethe ansatz has proved to be effective for spin chains of all classical types \cite{devegaExactBetheAnsatz1987,reshetikhinAlgebraicBetheAnsatz1988}, as well as Lie superalgebras \cite{martinsAlgebraicBetheAnsatz1997} and open spin chains \cite{belliardNestedBetheAnsatz2009} \cite{gerrardNestedAlgebraicBethe2020}. 

From the success of the nested Bethe ansatz for various $\g$, one might suppose the existence of a universal formulation of the nested Bethe ansatz, valid for all $\g$. 
There are hints that this might be possible: most prominently are the Bethe equations of the system, which depend only on the Cartan data of $\g$ \cite{ogievetskyFactorizedSmatrixBethe1986}.
Alas, a universal formulation of the nested Bethe ansatz seems, however, to fall at the first hurdle: the RTT relation from which the main commutation relations derive seems to rely specifically on the structure of the chosen $R$-matrix. 
In terms of the Yangian, the underlying quantum group \cite{sklyaninAlgebraicStructuresConnected1982, drinfeldQuantumGroups1988, jimboAqdifferenceAnalogueYangBaxter1985} from which rational spin chains derive their special properties, the difficulty lies in the rift between RTT presentation and Drinfel'd's presentations  \cite{drinfeldHopfAlgebrasQuantum1985, drinfeldNewRealizationYangians1987}. 
The fact that these presentations give an equivalent algebra was stated in \cite{drinfeldHopfAlgebrasQuantum1985} and proved rigorously in \cite{wendlandtRMatrixPresentationYangian2018}, but moving between the two presentations is nontrivial.
The work of Molev et al. \cite{molevYangiansClassicalLie2007, arnaudonRMatrixRealizationYangians2006} goes some way to connect the representation theoretic results of the two presentations, but is carried out in a type-by-type basis.

In this article we circumvent these difficulties and propose a new, universal approach to the nested Bethe ansatz.
The nested Bethe ansatz is generally implemented by ``cutting'' a $\g$-symmetric monodromy matrix into submatrices and, going through the appropriate steps, finding that the diagonal submatrices satisfy the relations of a subalgebra of type $\gs \subset \g$ with one simple root removed.
In our approach we work backwards -- we start by removing a simple root from $\g$.
This leaves a $U(1)$ charge and a diagram subalgebra $\gs \subset \g$ as expected; the ``cutting'' of the monodromy matrix follows simply from the decomposition of the auxiliary space into charge eigenspaces. 
In addition to the $U(1)$ charge, the other tool we make use of is the block Gauss decomposition of the $R$-matrix. 
With these two tools, we are able to deduce the appropriate properties of the ``vacuum sector'', the RTT relation for the diagonal blocks, and the AB commutation relation. 
Finally, we are also able to link the auxiliary space of the creation operator to a particular representation of $\gs$, which is consistent with the Bethe equations.

This paper is organised as follows. 
Section~\ref{s:algebra} is dedicated to introducing the required algebraic structures.
In Section~\ref{s:nesting} we introduce the idea of nesting as removing a simple root from $\g$; this is motivated by the Bethe equations
In Section~\ref{s:block-r}, we investigate the effect of nesting on the $R$-matrix, leading to our first main result in Proposition~\ref{p:nested-ybe}. 
In Section~\ref{s:monodromy}, we use the decomposition of the $R$-matrix to make deductions about the block relations of the monodromy matrix. 
The purpose of Section~\ref{sec:onex} is to connect these ideas to the construction of the eigenvector in the nested Bethe ansatz. 
We assert that the matrix creation operators act on an auxiliary space which has a particular representation, the properties of which are explored in Lemma~\ref{l:properties-of-auxsite}.
This leads to Conjecture~\ref{c:projector}, where we propose a relation between this representation and the $R$-matrix. 
Finally, we show how these ideas lead to the standard nested diagonalisation problem, and we recover the Bethe vector for a spin chain with classical $\g$ symmetry.

\section{Yangians for simple Lie algebras} \label{s:algebra}

\subsection{The Yangian}

Let $\g$ be a simple complex Lie algebra with Cartan subalgebra $\h$ of rank $r$. 
Denote the Killing form on $\h^*$ by $(\cdot,\cdot)$ and associated Cartan matrix by $A = (a_{ij})$. 
Let $\alpha_i \in \h^*$ for $1\leq i \leq r$ denote the simple roots of the Lie algebra and $\omega_i$ the corresponding fundamental weights, so that 
\[
	(\omega_i, \alpha_j) = d_j \delta_{ij} \qquad \text{and} \qquad (\alpha_i, \alpha_j) = a_{ij}d_j, \quad \text{where}\quad d_j = \tfrac{(\alpha_j,\alpha_j)}{2}
\]
We will use the same root labelling scheme as \cite{pressleyFundamentalRepresentationsYangians1991}.
Let $x^{\pm}_{\alpha}$ denote the root vector in $\g$ of root $\pm\alpha$, and abbreviate $x^{\pm}_{\alpha_i}=x^{\pm}_{i}$.
Define $h_{i} = [x^+_{i},x^-_{i}]$, so that the $h_{i}$ form a basis of $\h$.

We now introduce the Yangian of $\g$ in Drinfel'd's second presentation following closely the notation and definitions given in \cite{gautamPolesFinitedimensionalRepresentations2023}

\begin{definition} \label{def:yangian}
	Let $\hbar \in \C$. 
	The Yangian $Y_{\hbar}(\g)$ is the unital associative algebra generated by $h_{i,s}, x^{\pm}_{i,s}$ for $1 \leq i \leq r$, and $s \in \mathbb{Z}_{\geq 0}$, satisfying relations, for all $1 \leq i,j \leq r$, $s,t \in \mathbb{Z}_{\geq 0}$,
	\begin{gather}
		[h_{i,s}, h_{j,t}] = 0; \\
		[h_{i,0}, x^{\pm}_{j,t}] = \pm (\alpha_i,\alpha_j) x^{\pm}_{j,t}; \\
		[h_{i,s+1}, x^{\pm}_{j,t}] - [h_{i,s}, x^{\pm}_{j,t+1}]   = \pm \frac{\hbar}{2} (\alpha_i,\alpha_j) \left(
			h_{i,s} x^{\pm}_{j,t} + x^{\pm}_{j,t} h_{i,s}
		\right); \\
		[x^{\pm}_{i,s+1}, x^{\pm}_{j,t}] - [x^{\pm}_{i,s}, x^{\pm}_{j,t+1}]   = \pm \frac{\hbar}{2} (\alpha_i,\alpha_j) \left(
			x_{i,s} x^{\pm}_{j,t} + x^{\pm}_{j,t} x_{i,s}
		\right); \\
		[x^{+}_{i,s}, x^{-}_{j,t}] = \delta_{ij} h_{i,s+t};
	\end{gather}
	as well as the following relation, for $i \neq j, s_1, \dots, s_n, t \in \mathbb{Z}_{\geq 0}$ with $n=1-a_{ij}$,
	\begin{equation}
		\sum_{\pi \in S_{n}} \left[ x^{\pm}_{i,s_{\pi(1)}}, \left[ x^{\pm}_{i,s_{\pi(2)}}, \left[ \dots \left[x^{\pm}_{i,s_{\pi(m)}},x^{\pm}_{j,t}\right]\right]
		\right] \right] = 0.
	\end{equation}
	
\end{definition}

The Yangian may be endowed with a coproduct, counit and antipode, giving it the structure of a Hopf algebra. 
The formula for the coproduct for the Yangian in Drinfel'd's second presentation is not fully known. 
However, for the purposes of this article, we require only the following triangularity property, which we have taken from Proposition~2.9 of \cite{gautamPolesFinitedimensionalRepresentations2023}.
Write $Y^{\leq 0}$ for the subalgebra of $Y_{\hbar}(\g)$ with generators $\{h_{i,s}, x^-_{i,s}\}_{1 \leq i \leq r, s \in \mathbb{Z}}$, and write $Y^{\geq 0}$ for that generated by the $h_{i,s}, x^+_{i,s}$.
The adjoint action of $\h \subset Y_{\hbar}(\g)$ on $Y^{\geq 0}$ splits it into a direct sum of eigenspaces parametrised by the positive roots of $\g$, respectively the negative roots for $Y^{\leq 0}$. 
Labelling these spaces $Y^{\geq 0}_{\beta}$ for $\beta$ a positive root, the coproduct satisfies 
\begin{gather}
	\Delta(x^+_{i,s}) = x^+_{i,s} \ot 1 + 1 \ot x^+_{i,s} + \hbar \sum_{t=1}^s h_{i,t-1} \ot x^+_{i,s-t} \mod 
	\bigoplus_{\beta >0} Y^{\leq 0}_{-\beta} \otimes Y^{\geq 0}_{\beta};
	\\
	\Delta(x^-_{i,s}) = x^-_{i,s} \ot 1 + 1 \ot x^-_{i,s} + \hbar \sum_{t=1}^s x^-_{i,s-t} \ot h_{i,t-1} \mod 
	\bigoplus_{\beta >0} Y^{\leq 0}_{-\beta} \otimes Y^{\geq 0}_{\beta};
	\\
	\Delta(h_{i,s}) = h_{i,s} \ot 1 + 1 \ot h_{i,s} + \hbar \sum_{t=1}^s h_{i,s-t} \ot h_{i,t-1} \mod 
	\bigoplus_{\beta >0} Y^{\leq 0}_{-\beta} \otimes Y^{\geq 0}_{\beta}.
\end{gather}

Let us introduce formal power series 
\[
	x_i^{\pm}(u) = \sum_{s=0}^\infty x_{i,s}^{\pm}, \qquad h_i(u) = \sum_{s=0}^\infty h_{i,s}.
\]
In terms of these generating series, we introduce the shift automorphism $\tau_{a}$ for each $a \in \C$,
\[
	\tau_a(x_i^{\pm}(u)) = x_i^{\pm}(u-a) \qquad \tau_a(h_i(u)) = h_i(u-a).
\]

\subsection{Representations of the Yangian}

Here we list some facts about the representations of the Yangian. 
We restrict ourselves to finite-dimensional representations; it is understood that all finite-dimensional representations are \emph{highest weight representations}, defined as follows.
\begin{definition}
	A $Y_\hbar(\g)$-representation $V$ is said to be highest weight if there exists a vector $\xi \in V$ and formal power series $\mu_1(u), \dots, \mu_r(u) \in \C[[u^{-1}]]$ such that $V = Y_{\hbar}(\g)\cdot \xi$, and, for each $i = 1, \dots, r$,
	\[
		x^+_{i}(u) \cdot \xi = 0, \qquad h_{i}(u) \cdot \xi = \mu_{i}(u) \xi.
	\]
	In this case, the vector $\xi$ is called the highest weight vector of $V$ and $\mu_{i}(u)$ are known as the \emph{highest weights} of $V$.
\end{definition}

It was shown in \cite{drinfeldNewRealizationYangians1987} that the highest weight for an irreducible representation satisfies, for each $i=1, \dots, r$,
\[
	\mu_i(u) = \frac{P_i(u+d_i\hbar)}{P_i(u)},
\]
where $(P_1(u), \dots, P_r(u))$ are monic polynomials, called the \emph{Drinfel'd polynomials} of the representation.
The simplest nontrivial representations of $Y_\hbar(\g)$ are those for which 
\begin{equation}
	P_i(u) =
	\begin{cases}
		u-a & i=j \\
		1 &  \text{otherwise,}
	\end{cases}
\end{equation}
for some $a \in \C$ and some $j \in \{1, \dots, r\}$; this is known as the $j^\text{th}$ \emph{fundamental representation}, which we shall denote $M(\omega_j)_a$.
This notation is justified by the fact that the subspace generated by the action of $U(\g) \subset Y(\g)$ on the highest weight vector of the $j^\text{th}$ fundamental representation is equivalent to that of $\g$.
The conditions for this not to be a proper subspace -- that is, for a Yangian fundamental representation to be equivalent to the fundamental representation of $\g$ -- are given in \cite{drinfeldHopfAlgebrasQuantum1985}.
In general, for a $Y_{\hbar}(\g)$-highest weight representation with Drinfel'd polynomials $(P_1(u), \dots, P_r(u))$, the $\g$ weight of the highest weight vector is equal to $\deg P_1 \omega_1 + \dots + \deg P_r \omega_r$.

For any representation $V$ of $Y_{\hbar}(\g)$ with map $\rho:Y_{\hbar}(\g) \to \End(V)$, for any $a \in \C$ we may define another representation by the map $\rho \circ \tau_a$.
We will denote this representation by $V_a$.
If $V$ is irreducible with  Drinfel'd polynomials $(P_1(u), \dots, P_r(u))$, then so is $V_a$, with Drinfel'd polynomials $(P_1(u-a), \dots, P_r(u-a))$.

\subsection{The $R$-matrix and the extended Yangian}

Let $V$,$W$ be finite-dimensional irreducible representations of $Y_\hbar(\g)$.
For generic values of $a,b \in \C$, there exists an intertwiner $I^{VW}(a-b)$ of $Y_\hbar(\g)$ representations $W_b \ot V_a \to V_a \ot W_b$, which is a rational function in $(a-b)$ and is unique if its action on the tensor product of highest weight vectors is fixed to be unity, see \cite{drinfeldHopfAlgebrasQuantum1985}, also Theorem~4.1 of \cite{pressleyFundamentalRepresentationsYangians1991}.
Let $P$ denote the permutation matrix on the tensor product of two vector spaces, and let $R^{VW}(a-b) = I^{VW}(a-b) P$, a matrix in $\End(V \ot W)$, which satisfies the intertwining relation:
\begin{equation} \label{intertwining}
	R^{VW}(a-b) (\tau_a \ot \tau_b) \Delta^{\text{opp}}(x)
	=
	(\tau_a \ot \tau_b) \Delta(x) R^{VW}(a-b).
\end{equation}
We will refer to this is the $R$-matrix for $V$ and $W$. 
The $R$-matrix satisfies the \emph{unitarity relation}
\begin{equation} \label{unitarity}
	R^{VW}(a-b) R^{VW}(b-a) = I,
\end{equation}
and, taking $U$ to be another such representation, we have the \emph{quantum Yang-Baxter equation},
\begin{equation} \label{YBE}
	R^{UV}(a-b) R^{UW}(a-c) R^{VW}(b-c) = R^{VW}(b-c) R^{UW}(a-c) R^{UV}(a-b)
\end{equation}
in $\End(U \ot V \ot W)$.
We also have 
\begin{equation}\label{prp}
	R^{VV}(u) = P R^{VV}(u) P,
\end{equation}
where $P$ is the permutation matrix in $\End(V\ot V)$ satisfying $P (\zeta_1 \ot \zeta_2) = \zeta_2 \ot \zeta_1$, as well as $R^{VW}(u) \to I$ as $u \to \infty$.

It was stated in \cite{drinfeldHopfAlgebrasQuantum1985} and proved rigorously in \cite{wendlandtRMatrixPresentationYangian2018} that the $R$-matrix may be used to define an alternative presentation of the Yangian.

\begin{definition} \label{d:extended-Yangian}
	The extended Yangian $\X$ is the unital associative algebra over $\mathfrak{C}$ defined by generators $1, t^{(1)}_{ij}, t^{(2)}_{ij}, \dots$ for $1 \leq i,j \leq \dim(V)$ and relations 
	\begin{equation} \label{RTT}
		R_{12}^{VV}(u-v) T_1(u) T_2(v) = T_2(v) T_1(u) R_{12}^{VV}(u-v),
	\end{equation}
	where $T(u) = \sum_{ij} E_{ij} \otimes t_{ij}(u)$,
	and $t_{ij}(u) = 1 \delta_{ij} + t^{(1)}_{ij}u^{-1} + t^{(2)}_{ij}u^{-2} + \dots$, regarded as a formal series in $u^{-1}$.
	The extended Yangian may be endowed with a Hopf algebra structure by defining the \emph{coproduct} $\Delta$, \emph{antipode} $S$, and \emph{counit} $\epsilon$ by
	\begin{gather}
		\Delta(t_{ij}(u)) = \sum_{k} t_{ik}(u) \otimes t_{kj}(u),
		\\
		S\left(T(u)\right) = T(u)^{-1}
		\\
		\epsilon\left(T(u)\right) = I.
	\end{gather}
\end{definition}

The relationship between the Yangian and the extended Yangian was spelled out in detail in \cite{wendlandtRMatrixPresentationYangian2018}.
It was shown that, in general, 
\[
	\X \cong Y_{\hbar}(\g) \ot ZX_{\hbar}(\g),
\]
where $ZX_{\hbar}(\g)$ is the centre of $\X$, generated by the coefficients $z(u)$ in the relation $S^2(T(u))T(u+\tfrac12c_{\g}) = I\cdot z(u)$.
In $\X$, the shift automorphism $\tau_c$ is given by 
\begin{equation} \label{shift-automorphism}
	\tau_c : t_{ij}(u) \mapsto t_{ij}(u-c).
\end{equation}
Additionally, for any $f(u) = 1 + \sum_{s\geq 1} f^{(s)} u^{-s} + \dots$ we have the automorphism
\[
	T(u) \mapsto f(u) T(u) 
\]
From the defining relations of $\X$, it is clear that there is a correspondence between representations of $\X$ and solutions of the Yang-Baxter equation. 


The universal enveloping algebra $U(\g)$ appears as a Hopf subalgebra of $\X$,
\begin{equation}
	\phi : U(\g) \hookrightarrow \X
\end{equation}
This embedding may be characterised by $\phi \ot \rho_V (\Omega) \mapsto \hbar^{-1}T^{(1)}$, where $\Omega$ is the split Casimir element.
Hence, any representation of $\X$ defines a representation of $U(\g)$.
Taking the $u^{-1}$ coefficient of the RTT relation, we obtain
\[
	T_a^{(1)} T_b(v) + \hbar\, (\rho_{V_a} \ot \rho_{V_b})(\Omega) T_b(v) =  T_b(v) T_a^{(1)} +  \hbar\, T_b(v) (\rho_{V_a} \ot \rho_{V_b})(\Omega).
\]
In other words, 
\[
	(\rho_{V_a} \ot \rho_{V_b} \ot \phi) \circ (\Delta \ot 1)(\Omega) T_b(v) = T_b(v) (\rho_{V_a} \ot \rho_{V_b} \ot \phi)  \circ  (\Delta \ot 1)(\Omega) .
\]
Therefore, for any $x\in U(\g)$, 
\begin{equation} \label{g-yangian}
	[\phi(x) + \rho_{V_b}(x),T_b(v)] =0.
\end{equation}

\subsection{Fusion and spin chains}

Let $V_1$, $V_2$ be irreducible representations of the Yangian.
The intertwining relation \eqref{intertwining} implies that the $R$-matrix is a $\g$ intertwiner. 
Therefore, by Schur's lemma, it may be expanded in terms of projectors to the $\g$-irreducible components of $V_1 \otimes V_2$,
\[
	R^{V_1V_2}(u) = \sum_{W_k \subset V_1 \otimes V_2} g^{W_k}(u) \Pi^{W_k},
\]
where $g^{W_k}$ are rational functions. 
If we fix the limit $u \to \infty$, as $R^{V_1V_2}(u) \to I$, the functions are determined by their poles and zeros, which will collectively be referred to as singularities of the $R$-matrix. 
At these singularities (taking the residue in the case of a pole), the $R$-matrix fails to be invertible. 
Moreover, in some special cases the $R$ matrix will be proportional to a single projector, that is,
\[
	R^{V_1V_2}(z_1) \sim \Pi.
\]
Crucially, this projector still satisfies the Yang-Baxter equation, and so we have 
\begin{equation} \label{fused-YBE}
	\Pi^{V_1V_2} R^{V_1V_3}(u+z_1) R^{V_2V_3}(u) = R^{V_2V_3}(u) R^{V_1V_3}(u+z_1) \Pi^{V_1V_2}.
\end{equation}
We observe that $\Pi^{V_1V_2}$ may be multiplied on the left or the right without affecting the expression -- in other words, this expression preserves the image of $\Pi^{V_1V_2}$ in $V_1 \otimes V_2$.
Furthermore, writing $W$ for this image, we define
\begin{equation} 
	R^{W,V_3}(u) := \Pi^{V_1V_2} R^{V_1V_3}(u+z_1) R^{V_2V_3}(u)\Pi^{V_1V_2}
\end{equation}
Note that this is equal to \eqref{fused-YBE} as either of the projectors may be absorbed into the other using the Yang-Baxter equation.
In can be shown that this expression also satisfies a Yang-Baxter equation
\[
	R^{W,V_3}(u) R^{W,V_4}(u+v) R^{V_3 V_4}(v) =  R^{V_3 V_4}(v) R^{W,V_4}(u+v) R^{W,V_3}(u),
\]
where $V_4 \equiv V_3$. 
If a ``fusion point'' exists for $R^{V_3 V_4}(v)$ -- for example, if all $V_k$ are equivalent representations -- then the $V_3$ and $V_4$ spaces may be fused. 

A class of Yangian representations that may be constructed via fusion are the Kirillov-Reshetikhin representations \cite{kirillovRepresentationsYangiansMultiplicities1987}, which are obtained by fusion of certain fundamental representations.
These representations have Drinfel'd polynomials of degree $k\delta_{ij}$ for some $j$ and some non-negative $k$, satisfying
\[
	\frac{P_i(u+\hbar kd_i)}{P_i(u)} = \frac{u-c+\hbar kd_i/2}{u-c-\hbar kd_i/2}.
\]
The representations can be thought of as minimal affinisations of $\g$-modules with highest weight $k\omega_i$, generalising the fundamental representations. 
We will thus write these as $M(k\omega_i)_c$.

The quantum spin chain is simply a finite-dimensional $\X$-module $M$.
For the purposes of completeness of the Bethe ansatz, it is important that the module be irreducible.
Conditions for irreducibility of tensor product modules have been studied in \cite{tanLocalWeylModules2015}.

In principle, it is not necessary to specify the module further, as the techniques of the algebraic Bethe ansatz rely primarily on the algebraic relations of the Yangian. 
In practice, however, we choose $M$ to be a tensor product of Kirillov-Reshetikhin modules with shifts, 
\[
	M = M(\lambda_1)_{c_1} \ot \cdots \ot M(\lambda_L)_{c_L},
\]
where $\lambda$ is the highest Lie algebra weight of the module, and $c$ is a complex number. 
In this sense a spin chain may be specified by the data $(\g, M)$.
The representative of the Yangian generating matrix on the spin chain is referred to as the \emph{monodromy matrix}. 
Further, we define the \emph{transfer matrix} as the partial trace of the monodromy matrix over the auxiliary space. 
For full generality, we may introduce a matrix $Z \in \End(V)$ into the definition:
\[
	t(u) = \tr_a \left( Z_a T_a(u) \right) .
\]
In this sense, the transfer matrix is constructed from the data $(\g, V, M, Z)$.
If $Z \in G$, the compact, simply-connected Lie group associated with $\g$, then $Z \ot Z$ commutes with $R(u)$.
This gives 
\[
	[t(u), t(v)] = 0,
\]
from $RTT$ relation, using the fact that the $R$-matrix is invertible. 


The goal of the algebraic Bethe ansatz is to construct eigenvectors of these transfer matrices. 

\section{Nesting} \label{s:nesting}

\subsection{Review of the nested Bethe ansatz} \label{ssec:review}

In the algebraic Bethe ansatz, the construction of eigenvectors is achieved by acting with ``creation operators'' on a ``vacuum state'', the highest weight vector $\xi$ of the spin chain module. 
For the well-known case $\g = \mathfrak{sl}_2$, $V \cong \C^2$ the creation operator is taken simply from the upper-right element of the generating matrix. 
This gives the well-known expression for the Bethe vector,
\[
	t_{12}(v_1) t_{12}(v_2) \cdots t_{12}(v_m) \cdot \xi,
\]
which becomes an eigenvector when the $v_i$ satisfy the Bethe equations. 

In the nested Bethe ansatz, the generating matrix of the Yangian is partitioned into block matrices; the above- and below-diagonal blocks are treated as creation and annihilation block operators. 
\[
	T(u) = \left( \begin{matrix}
		A(u) & B(u) \\
		C(u) & D(u)
	\end{matrix} \right)
\]

The Bethe ansatz may then be performed in much the same way as for $\mathfrak{sl}_2$, with two caveats:
\begin{itemize}
	\item each creation operator acts on an additional vector space as well as the spin chain $M$, and
	\item there are multiple choices for the vacuum state as the block annihilation operators do not include all elements below the diagonal. 
\end{itemize} 
The end result is that the vacuum state must be chosen to be an eigenvector of the transfer matrix of a ``nested system'', which is a representation of a $Y(\mathfrak{k})$, a Yangian algebra with lower rank than the original.

In an analogous way to the original algebraic Bethe ansatz the spectral parameters of the block creation operators must satisfy the Bethe equations, which may be written in terms of the eigenvalue of the nested transfer matrix. 
Additionally, the diagonalisation of the nested transfer matrix leads to another set of Bethe equations. 
However, a more general picture of the Bethe ansatz \cite{ogievetskyFactorizedSmatrixBethe1986} shows that each set of Bethe equations is associated to one simple root of the original Lie algebra. 
The split Bethe equations which together solve the top level system and the nested system, and those which solve the original system must be equivalent. 
From this, we posit that $\mathfrak{k}$ must specifically be a diagram subalgebra of $\mathfrak{g}$, with one fewer simple root. 

Explicitly, the Bethe equations are, for $1 \leq i \leq r$,
\[
	\frac{P_i(v^{(i)}_k+\hbar d_i)}{P_i(v^{(i)}_k)} = -\prod_{j=1}^r \prod_{l=1}^{m^{(j)}} \frac{v^{(i)}_k-v^{(j)}_l+\tfrac{\hbar}{2}(\alpha_i,\alpha_j)}{v^{(i)}_k-v^{(j)}_l-\tfrac{\hbar}{2}(\alpha_i,\alpha_j)}.
\]
Suppose that the removed simple root of the nested system is $\alpha_p$. 
Then, the nested system must have Bethe equations, for $i\neq p$:
\[
	\frac{P_i(v^{(i)}_k+\hbar d_i)}{P_i(v^{(i)}_k)} = -\prod_{l=1}^{m^{(p)}} \frac{v^{(i)}_k-v^{(p)}_l+ \tfrac{\hbar}{2}(\alpha_i,\alpha_p)}{v^{(i)}_k-v^{(p)}_l-\tfrac{\hbar}{2}(\alpha_i,\alpha_p)} 
	\prod_{\substack{j=1 \\ j\neq p}}^r \prod_{l=1}^{m^{(j)}} \frac{v^{(i)}_k-v^{(j)}_l+ \tfrac{\hbar}{2}(\alpha_i,\alpha_j)}{v^{(i)}_k-v^{(j)}_l-\tfrac{\hbar}{2}(\alpha_i,\alpha_j)}.
\]
On the other hand, if we interpret this nested system as a spin chain system in its own right, then the roots $v^{(p)}$ do not feature as Bethe roots, they appear as shifts of the auxiliary spaces added to the spin chain.
In other words, they should be on ``the other side'' of the Bethe equations;
\[
	\frac{P_i(v^{(i)}_k+\hbar d_i)}{P_i(v^{(i)}_k)} \prod_{l=1}^{m^{(p)}} \frac{v^{(i)}_k-v^{(p)}_l-\tfrac{\hbar}{2}(\alpha_i,\alpha_p)}{v^{(i)}_k-v^{(p)}_l+ \tfrac{\hbar}{2}(\alpha_i,\alpha_p)}  = -
	\prod_{\substack{j=1 \\ j\neq p}}^r \prod_{l=1}^{m^{(j)}} \frac{v^{(i)}_k-v^{(j)}_l+ \tfrac{\hbar}{2}(\alpha_i,\alpha_j)}{v^{(i)}_k-v^{(j)}_l-\tfrac{\hbar}{2}(\alpha_i,\alpha_j)}.
\]
Using the fact that $(\alpha_i,\alpha_p)=-\max(d_p,d_i)=-kd_p$ if $i$ and $p$ label adjacent nodes on the Dynkin diagram of $\mathfrak{g}$, and zero otherwise, we expect that the Drinfel'd polynomials of the $l^\text{th}$ auxiliary space satisfy
\begin{equation}\label{nested-drin-polys}
	P_i(u+v_l^{(p)}) = 
	\begin{cases}
		(u-\tfrac12 \hbar d_p) & (\alpha_i,\alpha_p) = -d_p \\
		u(u- \hbar d_p) & (\alpha_i,\alpha_p) = -2d_p \\
		(u+\tfrac12 \hbar d_p)(u-\tfrac12 \hbar d_p)(u-\tfrac32 \hbar d_p) & (\alpha_i,\alpha_p) = -3d_p \\
		1 & \text{otherwise}
	\end{cases}.
\end{equation}
In all such cases, this defines a fundamental representation of the Yangian of $\mathfrak{k}$, since each connected component of the Dynkin diagram of the subalgebra only contains one such node. 
We will show how this representation manifests in the nested algebraic Bethe ansatz. 

\subsection{Nesting and the Yangian} \label{ssec:nesting-Yangian}

We now consider in more detail the subalgebra required for applying the nesting technique, as outlined in \ref{ssec:review}.
Choose a simple root $\alpha_p$ to exclude; we will restrict to the case where $\alpha_p$ is on the end of the Dynkin diagram, so is connected to only one other root.
Consider the fundamental co-weight $\omega^{\vee}_p$, satisfying 
$(\omega^{\vee}_p,\alpha_i) = \delta_{ip}$.
Denote its corresponding element in the Cartan subalgebra by $h^p \in \h$. 
In fact, it will slightly more convenient to focus on $-h^p$.
Then $[-h^p,x^{\pm}_{\alpha}]=\mp 2(\omega^{\vee}_p,\alpha) x^{\pm}_{\alpha}$ for any positive root $\alpha$.
We see that the adjoint representation of $\g$ decomposes into a direct sum of eigenspaces of $\mathrm{ad}_{(-h^p)}$ with eigenvalues $\{-n_p, \dots, n_p\}$; here $n_p$ is the coefficient of $-\alpha_p$ in the expansion of the highest root $\beta$ of $\g$ with respect to the basis of simple roots.
We will refer to $-h^p$ as the charge operator, and refer to its eigenvalues as charge.
Explicitly, let us write the decomposition into charge eigenspaces as 
\begin{equation} \label{adj-charge-decomp}
	\g \cong \g^{(-n_p)} \oplus \cdots \oplus \g^{(n_p)} .
\end{equation}
The charge zero space contains the Cartan subalgebra of $\g$, as well as all root vectors whose expansion in simple roots does not contain $\alpha_p$. In other words, 
\[
	\g^{(0)} \cong \gs \oplus \C (-h^p),
\]
where $\gs$ is a semisimple diagram subalgebra of rank $r-1$, with simple roots $\{\alpha_i\}_{i\neq p}$.
Since we restrict to the case where $\alpha_p$ is connected to only one other root, $\gs$ must be simple. 

Now consider the eigenspace decomposition of a finite-dimensional representation $V$ of $\g$ with respect to $(-h^p)$,
\begin{equation} \label{charge-decomp}
	V = \bigoplus_{J=0}^{N} V^J, \qquad (-h^p) V^J = c_J V^J =  (c_0 + J) V^J.
\end{equation}
Since $-h^p$ commutes with $\gs$, these eigenspaces may be regarded as representations of $\gs$.
Furthermore, any weight basis of $V$ respects the decomposition \eqref{adj-charge-decomp}, as the charge operator is contained in the Cartan subalgebra.
The space of maximal charge, $V^0$, contains the highest weight vector of $V$, and is irreducible as a representation of $\gs$ if $V$ is $\g$-irreducible.
The action of $\g^{(n)}$ increases the charge by $n$: 
\[
	x \cdot V^{J} \subset V^{J+n} \qquad \forall x \in \g^{(n)}.
\]
This defines a homomorphism of $\gs$ representations
\[
	\g^{(n)} \ot V^{J} \to V^{J+n}.
\]

The coproduct extends the notion of charge to the tensor product of representations, as
\[
	\Delta(-h^p) \cdot V^I \ot V^J = (c_I + c_J) V^I \ot V^J
\]
Therefore, the decomposition of the tensor product into charge eigenspaces follows 
\begin{equation} \label{tensor-charge-decomp}
	V \ot V = \bigoplus_{K} W_K \qquad \text{with} \qquad W_K = \bigoplus_{\substack{I,J \\ I+J=K}} (V^I \ot V^J).
\end{equation}

We apply the same decomposition to the $Y_{\hbar}(\g)$, which endows it with a grading according to charge.
We have from the defining relations,
\begin{equation*}
	\mathrm{ad}_{-h^p}(x^{\pm}_{i,s}) = \mp\delta_{ip}x^{\pm}_{i,s} \qquad \text{and} \qquad 
	\mathrm{ad}_{-h^p}(h_{i,s}) = 0,
\end{equation*}
mimicking the Lie algebra case.
Let $Y^{(n)}$ for $n \in \mathbb{Z}$ denote the subspace of charge $n$, spanned by monomials in Yangian generators whose total charge is equal to $n$. 
From the relations in Definition~\ref{def:yangian}, one may confirm that the generators $\{h_{i,s}, x^{\pm}_{i,s}\}_{i\neq p}$ generate a $Y_{\hbar}(\gs)$ subalgebra of charge zero when $r>2$. 
In the case $r=2$, so $\gs \cong \mathfrak{a}_1$, the relations give instead a $Y_{d\hbar}(\gs)$ subalgebra, where $d = (\alpha,\alpha)/2$, proportional to the length of the remaining root. 
In this case, the generators may be rescaled to recover $Y_{\hbar}(\mathfrak{a}_1)$; we will not discuss this case in detail here and will consider the subalgebra to be isomorphic to $Y_{\hbar}(\gs)$.

The subalgebra $Y_{\hbar}(\gs) \subset Y_{\hbar}(\g)$ is understood not to be a Hopf subalgebra, as noted in 2.18 of \cite{pressleyFundamentalRepresentationsYangians1991}.
However, we may use the triangularity of the coproduct to obtain the following result. 

\begin{lemma} \label{l:nested-coproduct}
	Let $\Delta^{\g}$ denote coproduct in $Y_{\hbar}(\g)$. Then, for any $x \in Y_{\hbar}(\gs) \subset Y_{\hbar}(\g)$, 
	\[
		\Delta^{\g} (x) = \Delta^{\gs} (x) \mod \bigoplus_{n\geq 1} Y^{(n)} \otimes Y^{(-n)}.
	\]
\end{lemma}

\begin{proof}
	The result follows by induction, using the same technique that was used to prove Proposition~2.9 in \cite{gautamPolesFinitedimensionalRepresentations2023}.
	We will reproduce the proof here for completeness.
	Indeed, let 
	\[
		t_{i,1} = h_{i,1} - \frac{\hbar}{2} h_{i,0}^2.
	\]
	It is understood that $Y_{\hbar}(\g)$ may be generated by $\{h_{i,0}, x^{\pm}_{i,0}, t_{i,1} \}_{1 \leq i \leq r}$, and we also have 
	\[
		\Delta^{\g}(t_{i,1}) = t_{i,1} \ot 1 + 1 \ot t_{i,1} - \hbar \sum_{\alpha \in (\Phi^\g)^+} (\alpha_i,\alpha) x^-_{\alpha} \ot x^+_{\alpha} .
	\]
	Observe that, for $i\neq p$, 
	\[
		\Delta^{\g}(t_{i,1}) - \Delta^{\gs}(t_{i,1})= - \hbar \sum_{\alpha \in (\Phi^\g)^+ \setminus (\Phi^\gs)^+} (\alpha_i,\alpha) \, x^-_{\alpha} \ot x^+_{\alpha} \in \bigoplus_{n\geq1} Y^{(n)} \ot Y^{(-n)}.
	\]
	For the inductive step, we have the following property: if we suppose $x, y \in Y_{\hbar}(\gs)$ satisfy the desired condition, then
	\begin{align*}
		\Delta^{\g}(xy)  &= \left(\Delta^{\gs}(x) + \bigoplus_{n\geq1} Y^{(n)} \ot Y^{(-n)} \right) \left( \Delta^{\gs}(y) + \bigoplus_{n\geq1} Y^{(n)} \ot Y^{(-n)} \right) 
		\\ 
		&= \Delta^{\gs}(x)\Delta^{\gs}(y) + \bigoplus_{n\geq1} Y^{(n)} \ot Y^{(-n)} 
		\\ 
		&= \Delta^{\gs}(xy) + \bigoplus_{n\geq1} Y^{(n)} \ot Y^{(-n)} .
	\end{align*}
	We have used the fact that $\Delta^{\gs}(x) \in Y^{(0)} \ot Y^{(0)}$.
	Inductively, this proves the result for all $x \in Y_{\hbar}(\gs)$.
\end{proof}

For what follows, it will be expedient to make the following assumption. 
\begin{center}
	\emph{Assume that $V$ is chosen such that $V^I$ are irreducible representations of $Y_\hbar(\gs)$. }
\end{center}
With this assumption, the charge spaces may be characterised by sets of Drinfel'd polynomials.
We wish to find these polynomials.
First, we can obtain the $\gs$-weight of each of the spaces $V^I$ by dropping the $p^\text{th}$ Dynkin label from their $\g$-weight.
The degree of the Drinfel'd polynomials can therefore be deduced, and it remains to find its roots.

For this purpose, let us introduce the normalised transfer operator $\overline{T}_i(u)$, as defined in \cite{gautamPolesFinitedimensionalRepresentations2023}.
We will not give the precise definition of $\overline{T}_i(u)$ here; rather, we will state some properties which will be useful later on. 
\begin{theorem}[Theorem~4.4 of \cite{gautamPolesFinitedimensionalRepresentations2023}]
	Let $V$ be a highest weight $Y_{\hbar}(\g)$ representation with highest weight vector $\xi$ and highest $\g$-weight $\lambda$. 
	Then for $i \in \{1, \dots, r\}$, there exists a map $\overline{T}_i(u) \in \End(V)[u]$ with properties as follows.
	\begin{itemize}
		\item Suppose $V_{\mu}$ is a $\g$-weight space of $V$ with $\g$-weight $\mu$. Then $\overline{T}_i(u)|_{V_\mu}$ is a monic polynomial of degree $(\lambda - \mu)(\omega^{\vee}_i)$. 
		\item Let $Q^{\g}_{i,V}(u)$ denote the eigenvalue of $\overline{T}_i(u)$ on the lowest weight space in $V$. Any eigenvalue of $\overline{T}_i(u)$ has eigenvalues that are monic polynomials which divide $Q^{\g}_{i,V}(u)$.
	\end{itemize}
\end{theorem}

This polynomial $Q^{\g}_{i,V}(u)$ is the \emph{$i^{\text{th}}$ Baxter polynomial} of $V$. 
With respect to our decomposition, the $\overline{T}_i(u)$ for $i\neq p$ are transfer operators for $\gs$ on each of the $V^I$, albeit with potentially the incorrect normalisation.
This implies that $Q^{\gs}_{i,V^I}(u)$ divides $Q^{\g}_{i,V}(u)$ for each $i \neq p$ and each $I$.
In particular, $P_i(u)$ divides $Q^{\gs}_{i,V^I}(u)$ for each of the Drinfel'd polynomials for $V^I$. 
In cases where $Q^{\gs}_{i,V^I}(u)$ does not have many factors -- for example, when all $V^I$ are fundamental or trivial representations -- this is sufficient to deduce $P_i(u)$.

\begin{example} \label{e:g2-subreps}
	Let $\g \cong \g_2$ and let the removed root be the long root, $\alpha_2$, so the subalgebra $\gs \cong \mathfrak{a}_1$.
	Let $V=M(\omega_1)$, the first fundamental representation.
	Then, from \cite{gautamPolesFinitedimensionalRepresentations2023} we obtain 
	\begin{align}
		Q^{\g}_{1,V}(u) &= (u-\tfrac12\hbar)(u-\tfrac52\hbar) \\
		Q^{\g}_{2,V}(u) &= u(u-2\hbar)(u-3\hbar)(u-5\hbar).
	\end{align}
	In terms of charge eigenspaces, $V$ decomposes as, in order of increasing charge,
	\[
		M^{\g_2}(\omega_1) \cong M^{\mathfrak{a}_1}(\omega) \oplus M^{\mathfrak{a}_1}(2\omega) \oplus M^{\mathfrak{a}_1}(\omega).
	\]
	Each of these can be extended to an irreducible representation of $Y_\hbar(\mathfrak{a}_1)$; it remains to find the shift parameter for each representation.
	For $V^0$, there is no shift, and the Drinfel'd polynomial is simply $P(u) = u$.
	As representations of $Y_\hbar(\mathfrak{a}_1)$, these representations have Baxter polynomials, which in the case of $\mathfrak{a}_1$ are simply the Drinfel'd polynomials, for some shifts $a_1, a_2$,
	\[
		Q^{\mathfrak{a}_1}_{1,V^1}(u) = (u-a_1)(u-a_1-\hbar)
		\qquad 
		\text{and}
		\qquad 
		Q^{\mathfrak{a}_1}_{1,V^2}(u) = (u-a_2).
	\]
	Comparison with the Baxter polynomials of $V$ yields $a_1=2\hbar$ and $a_2=5\hbar$.
	Hence, now as representations of the Yangian, 
	\[
		M^{\g_2}(\omega_1) \cong M^{\mathfrak{a}_1}(\omega) \oplus M^{\mathfrak{a}_1}(2\omega)_{2\hbar} \oplus M^{\mathfrak{a}_1}(\omega)_{5\hbar}.
	\]
\end{example}

Using similar techniques to Example~\ref{e:g2-subreps}, we have computed the shifts for each subalgebra in Table~\ref{t:v-decomp}.
We have restricted to the case where $V$ is a fundamental representation of the Yangian that is no larger than the corresponding fundamental representation of $\g$, and where each of the representations $V^I$ are fundamental or trivial representations; note that $\mathfrak{e}_8$ lacks such a representation.

\begin{table}
	
	\bgroup
	\def\arraystretch{1.3}

	\newcommand{\fun}[2]{M^{\gs}(\omega_{#1})_{#2}}
	\begin{tabular}{CCC|CCCCC}
		\g & \alpha_p & \gs  &  V & V^0 & V^1 & V^2 & V^3 
		\\
		\hline
		\hline
		\mf{a}_r & \alpha_1 & \mf{a}_{r-1} & M^{\g}(\omega_i) & \fun{i}{} & M^{\gs}(\omega_{i+1})_{\hbar/2} &
		\\
		\hline
		\multirow{2}{*}{$\mf{b}_r$}  & \multirow{2}{*}{$\alpha_1$} & \multirow{2}{*}{$\mf{b}_{r-1}$}
		& M^{\g}(\omega_1) & 1 & \fun{1}{\hbar} & 1
		\\
		&& & M^{\g}(\omega_r) & \fun{r-1}{} & \fun{r-1}{2\hbar} & 
		\\
		\hline
		\mf{c}_r & \alpha_1 & \mf{c}_{r-1} & M^{\g}(\omega_1) & 1 & \fun{1}{\hbar/2} & 1
		\\
		\mf{c}_r & \alpha_r & \mf{a}_{r-1} & M^{\g}(\omega_1) & \fun{1}{} & \fun{r-1}{(r+2)\hbar/2}
		\\
		\hline
		\multirow{2}{*}{$\mf{d}_r$} & \multirow{2}{*}{$\alpha_1$} & \multirow{2}{*}{$\mf{d}_{r-1}$} & M^{\g}(\omega_1) & 1 & \fun{1}{\hbar/2} & 1 
		\\
		& & & M^{\g}(\omega_{r-1}) & \fun{r-2}{} & \fun{r-1}{\hbar} & 
		\\
		\mf{d}_r & \alpha_{r} & \mf{a}_{r-1} & M^{\g}(\omega_1) & \fun{1}{} & \fun{r-1}{(r-2)\hbar/2} & 
		\\
		\hline
		\multirow{2}{*}{$\mf{e}_6$} & \multirow{2}{*}{$\alpha_6$} & \multirow{2}{*}{$\mf{d}_5$} & M^{\g}(\omega_1) & \fun{1}{} & \fun{5}{3\hbar/2} & 1
		\\
		& &  & M^{\g}(\omega_6) & 1 & \fun{4}{\hbar/2} & \fun{1}{2\hbar}
		\\
		\mf{e}_6 & \alpha_2 & \mf{a}_5 & M^{\g}(\omega_1) & \fun{1}{} & \fun{4}{3\hbar/2} & \fun{1}{3\hbar}
		\\
		\hline
		\mf{e}_7 & \alpha_7 & \mf{e}_6 & M^{\g}(\omega_7) & 1 & \fun{6}{\hbar/2} & \fun{1}{5\hbar/2} & 1
		\\
		\mf{e}_7 & \alpha_1 & \mf{d}_6 & M^{\g}(\omega_7) & \fun{1}{} & \fun{5}{2\hbar} & \fun{1}{4\hbar}
		\\
		\mf{e}_7 & \alpha_2 & \mf{a}_6 & M^{\g}(\omega_7) & \fun{6}{} & \fun{2}{2\hbar} & \fun{5}{7\hbar/2} & \fun{1}{11\hbar/2}
		\\
		\hline
		\mf{f}_4 & \alpha_1 & \mf{c}_3 & M^{\g}(\omega_4) & \fun{1}{} & \fun{2}{5\hbar/2} & \fun{1}{5\hbar} &
		\\
		\hline
		\mf{g}_2 & \alpha_2 & \mf{a}_1 & M^{\g}(\omega_1) & M^{\gs}(\omega) & M^{\gs}(2\omega)_{2\hbar} & M^{\gs}(\omega)_{5\hbar}
	\end{tabular}

	\egroup

	\caption{Decomposition of $V$ into representations of the Yangian. }
	\label{t:v-decomp}
\end{table}

\section{Block decomposition of the \texorpdfstring{$R$}{R}-matrix} \label{s:block-r}

Let us now investigate the effect of the decomposition into charge eigenspaces on $R(u)$.
Since the $R$ matrix is $\g$-invariant, it is clearly $(\gs \oplus \mathbb{C})$-invariant. 
Hence, it respects the decomposition into total charge eigenspaces \eqref{tensor-charge-decomp}: it is block diagonal with respect to this decomposition. 
This matrix decomposes further into blocks which map between tensor spaces $V^I \ot V^J$ of equal charge:
\[
	R^{IJ}_{KL}(u) \in \Hom_{\gs}(V^K \ot V^L,V^I \ot V^J),
\]
where $I+J=K+L$.

It is immediate from the above that $R^{00}_{00}(u)$ and $R^{NN}_{NN}(u)$ satisfy Yang-Baxter equations on $V^0 \ot V^0 \ot V^0$ and $V^N \ot V^N \ot V^N$ respectively -- this follows from conservation of charge. 
In general, however, the matrix blocks of the $R$-matrix defined above are not $R$-matrices in their own right. 
Nevertheless, there must exist Yangian intertwiners which intertwine these spaces.
In what follows, we reconcile these ideas using the block Gauss decomposition of the $R$-matrix.

Define an ordering on spaces $V^I \ot V^J$, ordering them first by total charge, then by the label of the first tensor factor, so
\[
	V^0 \ot V^0 < V^0 \ot V^{1} < V^{1} \ot V^0  < V^0 \ot V^{2} < \dots
\]

\begin{proposition} \label{UDL-decomp}
	For generic values of $u\in \C$, the $R$-matrix admits a unique block decomposition
	\begin{equation} 
		R(u) = U(u) D(u) L(u), \label{r-udl}
	\end{equation}
	where $D(u)$ is block diagonal, and $U(u)$ and $L(u)$ are block upper- and lower-triangular respectively with diagonal blocks equal to the identity matrix.
	The blocks are defined recursively as follows: 
	\begin{equation} \label{d-def}
		D^{IJ}(u) = R^{IJ}_{IJ}(u) - \sum_{k=1}^{\min(\Jmax-I,J-1)} U^{IJ}_{I+k, J-k}(u) D^{I+k, J-k}(u) L^{I+k, J-k}_{IJ}(u),
	\end{equation}
	\begin{multline} \label{u-def}
		U^{IJ}_{I+j,J-j}(u) = \left[R^{IJ}_{I+j,J-j}(u) - \sum_{k=j+1}^{\min(\Jmax-I,J-1)} U^{IJ}_{I+k, J-k}(u) D^{I+k, J-k}(u) L^{I+k, J-k}_{I+j,J-j}(u) \right] \\ \times [D^{I+j,J-j}(u)]^{-1},
	\end{multline}
	and
	\begin{equation} \label{l-def}
		L^{I+j,J-j}_{IJ}(u) = [D^{I+j,J-j}(u)]^{-1} \left[R^{I+j,J-j}_{IJ}(u) - \sum_{k=1}^{j-1} U^{I+j,J-j}_{I+k, J-k}(u) D^{I+k, J-k}(u) L^{I+k, J-k}_{IJ}(u) \right].
	\end{equation}
\end{proposition}

\begin{proof}
	This can be obtained using e.g. \cite{stanimirovicFullrankBlockLDL2012}.

	The validity of the decomposition relies on the invertibility of the matrices $D^{IJ}$.
	The determinant of $D^{IJ}$ is a rational function in $u$, from the rationality of the $R$-matrix.
	As $u\to \infty$, the $R$-matrix tends to the identity matrix, hence the decomposition is valid in some open region containing infinity, and so none of the $\det(D^{IJ}(u))$ are identically equal to zero. 
\end{proof}

The matrices $U(u)$ and $L(u)$ are invertible due to being upper/lower triangular with 1s on the diagonal, and $D(u)$ is assumed to be invertible for generic values of $u$. 
By taking the inverse of the above decomposition, and using the unitarity of the $R$-matrix \eqref{unitarity}, we obtain the following. 

\begin{corollary}\label{cor:ldu-decomp}
	The $R$-matrix admits a unique block decomposition
	\begin{equation} 
		R(u) = \widetilde{L}(u) \widetilde{D}(u) \widetilde{U}(u), \label{r-ldu}
	\end{equation}
	where $\widetilde{D}(u)$ is block diagonal, and $\widetilde{U}(u)$ and $\widetilde{L}(u)$ are block upper- and lower-triangular respectively with diagonal blocks equal to the identity matrix.
	The blocks are related to those in Proposition~\ref{UDL-decomp} by 
	\begin{equation}
		\widetilde{D}(u) = D(-u)^{-1}, \qquad 
		\widetilde{L}(u) = U(-u)^{-1}, \qquad
		\widetilde{U}(u) = L(-u)^{-1}.
	\end{equation}
\end{corollary}

We list below some practical relations satisfied by the $D$-matrices.

\begin{lemma} \label{d-matrices}
	For $1 \leq I,J \leq \Jmax$,
	\begin{gather}
		D^{J1}(u) = R^{J1}_{J1}(u) \label{DJ1}
		\\
		D^{\Jmax J}(u) = R^{\Jmax J}_{\Jmax J}(u) \label{DNJ}
		\\
		D^{1J}(u)=[R^{1J}_{1J}(-u)]^{-1} \label{D1J}
		\\
		D^{J\Jmax}(u)=[R^{J\Jmax}_{J\Jmax}(-u)]^{-1} \label{DJN}
		\\
		P\,D^{IJ}(u)\, P = \left[D^{JI}(-u)\right]^{-1} \label{PDP}
	\end{gather}
\end{lemma}

\begin{proof}
	Equations \eqref{DJ1}\eqref{DNJ} hold from the definition, while \eqref{D1J}\eqref{DJN} are obtained as a consequence of \eqref{PDP} and \eqref{prp}.
	It remains to prove \eqref{PDP}.
	Note that the triangular matrices have identity matrices on the diagonal, and are thus invertible with triangular inverse. 
	Hence, by unitarity \eqref{unitarity}, we have
	\begin{equation} \label{r-1-ldu}
		R(u) = R(-u)^{-1} = L(-u)^{-1}D(-u)^{-1}U(-u)^{-1}.
	\end{equation}
	Alternatively, conjugating \eqref{r-udl} by the permutation operator, and applying \eqref{prp} we have 
	\[
		R(u) = P\,R(u)\, P = \widetilde{L}(u) \widetilde{D}(u) \widetilde{U}(u),
	\]
	where $\widetilde{L}(u) = P\, U(u) P$ is block lower triangular with identity matrices on the diagonal,
	$\widetilde{U}(u) = P\, U(u) P$ is block upper triangular with identity matrices on the diagonal.
	Equating this with \eqref{r-1-ldu}, we may obtain the equality
	\[
		L(-u) \widetilde{L}(u) \widetilde{D}(u) = D(-u)^{-1}U(-u)^{-1} \widetilde{U}(u)^{-1}.
	\]
	The left- and right-hand sides are block upper- and lower-triangular respectively, and hence both sides are diagonal. 
	Then, equating the diagonal blocks, only the $D$-matrices remain. 
	Noting that the $V^I \ot V^J$ block of $\widetilde{D}(u) = P\, D(u) P$ is equal to $D^{JI}(u)$, we obtain the result. 
\end{proof}


Note that in the case $N=1$, which encompasses most known cases of the nested Bethe ansatz, the results of Lemma~\ref{d-matrices} are enough to obtain $D(u)$ from matrix blocks of $R(u)$ without the need to resort to any complex calculations.
In the general case, the importance of these $D$-matrices is borne out in the following result.

\begin{proposition} \label{p:nested-ybe}
	The matrices $\widetilde{D}^{IJ}(u)$ satisfy the intertwining equation 
	\begin{equation} \label{d-intertwining}
		\widetilde{D}^{IJ}(a-b) (\tau_a \ot \tau_b) \Delta^{\text{opp}}(x) \cdot (V^{I} \ot V^{J})
		=
		(\tau_a \ot \tau_b) \Delta(x) \widetilde{D}^{IJ}(a-b) \cdot (V^{I} \ot V^{J})
	\end{equation}
	for $x \in Y_\hbar(\gs) \subset Y_\hbar(\g)$.
	In particular, the matrices $\widetilde{D}^{IJ}(u)$ and $D^{IJ}(u)$  satisfy the Yang-Baxter equations
	\begin{gather}
		\widetilde{D}^{IJ}(u) \widetilde{D}^{IK}(u+v) \widetilde{D}^{JK}(v)
		= \widetilde{D}^{JK}(v) \widetilde{D}^{IK}(u+v) \widetilde{D}^{IJ}(u)
		\\
		{D}^{IJ}(u) {D}^{IK}(u+v) {D}^{JK}(v)
		= {D}^{JK}(v) {D}^{IK}(u+v) {D}^{IJ}(u)
	\end{gather}
	on $V^I \ot V^J \ot V^K$.
\end{proposition}

\begin{proof}
	Let $x \in Y_{\hbar}(\gs) \subset Y_{\hbar}(\g)$. 
	Our starting point is the intertwining equation for the $R$-matrix \eqref{intertwining},
	\[
		R^{VV}(a-b) (\tau_a \ot \tau_b) \Delta^{\g, \text{opp}}(x)
		=
		(\tau_a \ot \tau_b) \Delta^{\g}(x) R^{VV}(a-b).
	\]
	From Lemma~\ref{l:nested-coproduct}, we have that the coproduct is block lower triangular,
	\[
		\Delta^{\g}(x) \subset \bigoplus_{n\geq 0} V^{I+n} \ot V^{J-n},
	\]
	and the diagonal blocks of the coproduct agree with $\Delta^{\gs}(x)$.
	Similarly, $\Delta^{\g,\text{opp}}(x)$ is block upper triangular; this can be shown using the Cartan involution.
	
	Now, applying the $LDU$ decomposition of the $R$-matrix, we have 
	\[
		\widetilde{L}(a-b)\widetilde{D}(a-b)\widetilde{U}(a-b) (\tau_a \ot \tau_b) \Delta^{\g, \text{opp}}(x)
		=
		(\tau_a \ot \tau_b) \Delta^{\g}(x) L(a-b)D(a-b)U(a-b).
	\]
	Next, we multiply on the left and right by $\widetilde{L}(a-b)^{-1}$ and $\widetilde{U}(a-b)^{-1}$ respectively, to obtain 
	\[
		\widetilde{D}(a-b)\widetilde{U}(a-b) (\tau_a \ot \tau_b) \Delta^{\g, \text{opp}}(x) \widetilde{U}(a-b)^{-1}
		=
		\widetilde{L}(a-b)^{-1} (\tau_a \ot \tau_b) \Delta^{\g}(x) \widetilde{L}(a-b)\widetilde{D}(a-b).
	\]
	From this position, we observe that the left-hand side of the equation is entirely block upper triangular, while the right-hand side is block lower triangular. 
	We therefore conclude that both sides are in fact block diagonal. 
	These diagonal blocks can be obtained by multiplying together the diagonal blocks of each side respectively.
	For the $\widetilde{L}$ and $\widetilde{U}$ matrices, and their inverses, this simply gives the identity matrix, resulting in \eqref{d-intertwining}.

	The Yang-Baxter equation for $\widetilde{D}(u)$ follows from the uniqueness of the $R$-matrix.
	The transformation $\widetilde{D}(u) \mapsto \widetilde{D}(-u)^{-1} = D(u)$ preserves the Yang-Baxter equation.
\end{proof}



\begin{remark}
		The Gauss (triangular) decomposition of the Yangian generating matrix has been used to connect the RTT presentation of the Yangian to its current presentation \cite{khoroshkinYangianDouble1996} \cite{khoroshkinGaussDecompositionTrigonometric1995} \cite{liashykBetheVectorsOrthogonal2019}.
		In that case, the equivalent $D^{IJ}$ matrices form the commutative Gelfand-Tsetlin subalgebra of the Yangian. 
		In this sense, the \emph{block} Gauss decomposition may be interpreted as a convenient presentation for the nested Bethe ansatz which somehow lies between the RTT presentation and the Drinfel'd current presentation. 
\end{remark}

Missing from the above is a method for determining the normalisation of $D^{IJ}(u)$.
In the case of $D^{I0}(u)=R^{I0}_{I0}(u)$, if $V$ is a fundamental representation of $Y_{\hbar}(\g)$ we may obtain the normalisation factor using Proposition~7.7 of \cite{gautamPolesFinitedimensionalRepresentations2023}. This gives the diagonal matrix element for an $R$ matrix acting on a highest weight vector (this is often referred to as a ``highest weight'' when working in the context of the RTT presentation).

\begin{example}
	Choose $\g = \g_2$, and $V$ the first fundamental representation. 
	From the full Baxter polynomial, we find
	\[
		Q^{\g_2}_{1,v^{1}}(u) = u;
		\qquad 
		Q^{\g_2}_{1,v^{2}}(u) = u(u-2\hbar)(u-3\hbar).
	\]
	Therefore, the normalisation factor is 
	\[
		f^{01}(u) = \frac{u}{u-\hbar}
		\qquad 
		f^{02}(u) = \frac{u(u-2\hbar)\cancel{(u-3\hbar)}}{(u-\hbar)\cancel{(u-3\hbar)}(u-4\hbar)}
		= \frac{u}{u-\hbar} \cdot\frac{u-2\hbar}{u-4\hbar}.
	\]
\end{example}

\section{Block decomposition of the monodromy matrix} \label{s:monodromy}

Under the decomposition \eqref{charge-decomp}, and choosing a basis for $V$ which respects the decomposition, the Yangian generating matrix decomposes into blocks $T^{I}_{J}(u) \in \Hom(V^J,V^I) \ot Y(\g)$.
In keeping somewhat with the standard algebraic Bethe ansatz notation, we write $A^{I}_I(u)$ for diagonal blocks, $B^{I}_{J}(u)$ for above-diagonal blocks, so $I<J$, and  $C^{J}_{I}(u)$ for below-diagonal blocks.
\[
	T(u)  = \left(
	\begin{matrix}
		A^0_0(u) & B^{0}_{1}(u) & \cdots & B^{0}_{N}(u) \\
		C^{1}_{0}(u) & A^{1}_{1}(u) & \cdots & B^{1}_{N}(u) \\
		\vdots & \vdots & \ddots & \vdots \\
		C^{N}_{0}(u) & C^{N}_{1}(u) & \cdots & A^{\Jmax}_{\Jmax}(u) \\
	\end{matrix}
	\right)
\]
The commutation relations between the blocks may be deduced from the $RTT$ relation \eqref{RTT}, and the decomposition of the $R$ matrix.

\begin{lemma} \label{l:grading}
	The charge operator institutes a $\mathbb{Z}$-grading on the elements of $\X$, given by 
	\[
		(1 \ot \mathrm{ad}_{-h^p})(T^{I}_{J}(u)) = -(I-J) T^{I}_{J}(u).
	\]
\end{lemma}

\begin{proof}
	Let $\rho_V$ denote the map $X_{\hbar}(\g) \to \End(V)$, and $\phi$ the embedding $U(\g) \hookrightarrow \X$.
	Using the fact that the Yangian generating matrix is $\g$-invariant \eqref{g-yangian}, the charge operator satisfies 
	\[
		[\phi(-h^p) + \rho_V(-h^p),T(u)] = 0.
	\]
	Taking the matrix block corresponding to $\Hom(V^J,V^I)$, we conclude that the elements of the Yangian generating matrix are eigenvectors of the adjoint action of the charge operator in such a way that the eigenvalue is exactly the negative of the Lie algebra charge. 
\end{proof}

As with the decomposition of the $R$ matrix, we can make some simple deductions in the cases where the blocks are defined on subspaces of extremal charge. 

\begin{corollary} \label{C:Y-extremal-charge}
	The following relations are satisfied by the Yangian blocks
	\begin{gather*}
		R^{00}_{00}(u-v) (A^{0}_{0}(u))_a (A^{0}_{0}(v))_b =  (A^{0}_{0}(v))_b  (A^{0}_{0}(u))_a R^{00}_{00}(u-v)
		\\
		R^{NN}_{NN}(u-v) (A^{N}_{N}(u))_a (A^{N}_{N}(v))_b =  (A^{N}_{N}(v))_b  (A^{N}_{N}(u))_a R^{NN}_{NN}(u-v)
	\end{gather*}
\end{corollary}

The $R$-matrix here satisfies the Yang-Baxter equation \eqref{YBE} so, noting also that, when regarded as a series in $u^{-1}$, the constant term in $A^{I}_{I}(u)$ is equal to the identity matrix, this immediately gives us a homomorphism of associative algebras $\Xs \hookrightarrow \X$. 
This, however, is not necessarily a homomorphism of Hopf algebras: for one, the restriction of the coproduct does not give a valid coproduct for the subalgebra. 
To remedy this, as well as prove further results, we make the following construction. 

Let $\mcC$ denote the right ideal generated by $C$ operators, that is, $\mcC := \bigcap_{k>l,s,i,j} \X [C^{kl}]^{(s)}_{ij}$.
We then consider $\X$ itself as a left $\X$-module, in which case $\mcC$ is an invariant submodule, and we may form the quotient module $\X/\mcC$. 
Alternatively, we can define the \emph{idealiser} $\mcI_{\X}(\mcC)$ of $\mcC$ as the subspace of elements $x$ of $\X$ for which $\mcC x = \mcC$. 
This subspace is closed under multiplication, and so is a subalgebra of the $\X$.
Moreover, $\mcC$ is a two-sided ideal in $\mcI_{\X}(\mcC)$, and so we may define $\mcI_{\X}(\mcC)/\mcC$.
We will consider both of these constructions in what follows. 

\begin{proposition}\label{P:DAA}
	For any $I,J$, we have
	\[
		D_{ab}^{IJ}(u-v) (A^{I}_{I}(u))_a (A^{J}_J(u))_b \equiv (A^{J}_J(u))_b (A^{I}_{I}(u))_a D_{ab}^{IJ}(u-v) + \mcC.
	\]
	Moreover, for each $I$ where $V^I$ is non-trivial, the maps
	\begin{gather*}
		X_{\hbar}(\gs) \to \mcI_{\X}(\mcC)/\mcC \\
		T(u) \mapsto A^{I}_{I}(u) 
	\end{gather*} 
	are bialgebra homomorphisms.
\end{proposition}

\begin{proof}
	The strategy is identical to that of Proposition~\ref{p:nested-ybe}.

	First, we consider the RTT relation \eqref{RTT}.
	With respect to the decomposition of $V \ot V$ into spaces of equal total charge, the $R$ matrix is block diagonal, and so we may consider the both sides of the RTT relation block-wise.
	Consider the diagonal blocks of $T_a(u)T_b(v)$; breaking these blocks further to blocks which map between the spaces $V^I \ot V^J$, these blocks have the form
	\[
		\begin{pmatrix}
			(A^{0}_{0}(u))_{a} (A^{I}_{I}(v))_{b} & (B^{0}_{1}(u))_{a} (C^{I}_{I-1}(v))_b & (B^{0}_{2}(u))_{a} (C^{I}_{I-2}(v))_b & \cdots 
			\\
			(C^{1}_{0}(u))_{a} (B^{I-1}_{I}(v))_b & (A^{1}_{1}(u))_{a} (A^{I-1}_{I-1}(v))_b  & (B^{1}_{2}(u))_{a} (C^{I-1}_{I-2}(v))_b &
			\\
			(C^{2}_{0}(u))_{a} (B^{I-2}_{I}(v))_b & & \ddots  &
			\\
			\vdots & &&
		\end{pmatrix}.
	\]
	Likewise, $T_b(v)T_a(u)$ may be written as
	\[
		\begin{pmatrix}
			(A^{I}_{I}(v))_{b} (A^{1}_{1}(u))_{a} & (C^{I}_{I-1}(v))_b (B^{1}_{2}(u))_{a}  &  (C^{I}_{I-2}(v))_b (B^{1}_{3}(u))_{a} & \cdots 
			\\
			(B^{I-1}_{I}(v))_b (C^{2}_{1}(u))_{a} &  (A^{I-1}_{I-1}(v))_b (A^{2}_{2}(u))_{a} & (C^{I-1}_{I-2}(v))_b (B^{2}_{3}(u))_{a}  &
			\\
			(B^{I-2}_{I}(v))_b (C^{3}_{1}(u))_{a} & & \ddots  &
			\\
			\vdots & &&
		\end{pmatrix}.
	\]
	Now we apply the block Gauss decomposition for the $R$-matrix from Proposition~\ref{UDL-decomp},
	\[
		R(u) = U(u) D(u) L(u).
	\]
	Since $L(u)$ is generically invertible, with block lower-triangular inverse, we may write
	\begin{align*}
		U(u-v) D(u-v) L(u-v) T_a(u) T_b(v) &=  T_b(v)T_a(u) U(u-v) D(u-v) L(u-v)
		\\
		\Rightarrow  D(u-v) L(u-v) T_a(u) T_b(v) \left[L(u-v)\right]^{-1}&=  \left[U(u-v)\right]^{-1} T_b(v)T_a(u) U(u-v) D(u-v).
	\end{align*}
	We see that in $\mcI_{\X}(\mcC)/\mcC$, the left-hand side is block lower triangular, while the right-hand side is block upper triangular. Hence, both sides are block diagonal, and we may deduce their values by the product of diagonal blocks. 
	Since the diagonal blocks of the $L(u)$ and $U(u)$ matrices are identity matrices, we obtain 
	\begin{equation*} 
		D^{IJ}(u-v) (A^{I}_{I}(u))_a (A^{J}_{J}(v))_b \equiv (A^{J}_{J}(v))_b (A^{I}_{I}(u))_a D^{IJ}(u-v),
	\end{equation*}
	where $\equiv$ denotes equivalence in $\mcI_{\X}(\mcC)/\mcC$.
	Finally, the $D^{IJ}(u)$ matrix inherits its rational dependence on $u$ from the $R$ matrix.

	To show that this is a closed algebra, we employ the Poincar{\'e}-Birkhoff-de Witt theorem for the extended Yangian from \cite{wendlandtRMatrixPresentationYangian2018}.
	Choose an ordering on Yangian generators which respects the charge, so lower charge elements are lower in the ordering. 
	Then, by the PBW theorem, any element of the extended Yangian may be written as an ordered product of these generators. 
	Matrix elements of the $C$ matrices have strictly positive charge by Lemma~\ref{l:grading}, and so, up reordering, $\mcC$ must contain all elements of $\X$ with strictly negative charge.
	In particular, $\mcC a \in \mcC$ for any $a$ with zero charge, implying $a \in \mcI_{\X}(\mcC)$.
	Finally, due to the relation,
	\[
		D_{ab}^{II}(u-v) (A^{I}_{I}(u))_{a} (A^{I}_{I}(u))_{b} \equiv (A^{I}_{I}(u))_{b} (A^{I}_{I}(u))_{a} D_{ab}^{II}(u-v) + \mcC,
	\]
	we see that the matrix elements of $A^{I}_{I}(u)$ form a closed subalgebra in $\mcI_{\X}(\mcC)$.
	
	It remains to show the bialgebra structure. 
	The matrix $T(u)$ becomes block upper triangular in  $\mcI_{\X}(\mcC)/\mcC$, and so
	\begin{multline}
		\Delta(A^{I}_{I}(u)) = \sum_{J} (T^{I}_{J}(u))_{[1]} \cdot (T^{J}_{I}(u))_{[2]} 
		= (A^{I}_{I}(u))_{[1]} \cdot (A^{I}_{I}(u))_{[2]}
		\\
		+ \sum_{J<I} (C^{I}_{J}(u))_{[1]} \cdot (B^{J}_{I}(u))_{[2]}
		+ \sum_{J>I} (B^{I}_{J}(u))_{[1]} \cdot (C^{J}_{I}(u))_{[2]} ,
	\end{multline}
	where $\cdot$ denotes matrix multiplication in the auxiliary space.
	Therefore, in the quotient algebra,
	\[
		\Delta(A^{I}_{I}(u)) \mapsto (A^{I}_{I}(u))_{[1]} \cdot (A^{I}_{I}(u))_{[2]} + \mcC_{[1]} + \mcC_{[2]}.
	\] 
	Finally, the counit trivially descends to a counit on $\mcI_{\X}(\mcC)/\mcC$.
\end{proof}

By taking the trace of the relation above, as the $D^{IJ}(u)$ are invertible, we immediately obtain the following. 

\begin{corollary} \label{cor:commuting-transfer}
	Denote equivalence in $\mcI_{\X}(\mcC)/\mcC$ by $\equiv$. 
	Suppose $Z^I \ot Z^J$ commutes with $D^{IJ}(u-v)$ for any two $I$ and $J$. 
	Then,
	\[
		\tr_I Z^I A^{I}_{I}(u) \tr_J Z^J A^{J}_{J}(v) \equiv \tr_J Z^J A^{J}_{J}(v) \tr_I Z^I A^{I}_{I}(u)
	\]
\end{corollary}

Since we have shown the RTT relation above for all $A$, we can consider acting on an irreducible representation $M$ of $\X$. 
First, we define the \emph{vacuum sector} $\ker \mathcal{C} = \bigcap_{C \in \mcC} \ker C \subset M$, the space annihilated by all $C$ operators.
It is analogous to the term ``vacuum vector'', which is another name for the highest-weight vector of $M$ that is often used in theoretical physics. 
In the nested Bethe ansatz, the vacuum sector plays the same role that the vacuum vector does in the algebraic Bethe ansatz.
From Proposition~\ref{P:DAA}, we immediately have the following result. 

\begin{lemma} \label{l:vac-sec-highest-charge}
		The vacuum sector is equal to the charge zero subspace $M^0$ of $M$  and is invariant under the action of $A^{I}_{I}(u)$ for each $I$.
\end{lemma}

\begin{proof}
	The invariance under the action of $A$ operators is immediate from Proposition~\ref{P:DAA}.

	To show the second part, note first that the matrix $T(u)$ conserves the total charge of a vector in $V \ot M$. Then, since each $C^{I}_{J}(u)$ raises the charge in $V$ by $(I-J)$, it must lower the charge in $M$ by an equivalent amount. 
	Hence, each $C^{I}_{J}(u)$ annihilates the space of charge zero, and so $\ker \mathcal{C} = M^0$

	Now suppose we have $\eta \in \ker \mathcal{C} $.
	As $M$ is irreducible, there exists a monomial in generators such that 
	\[
		\xi \propto t^{(s_1)}_{k_1l_1} \cdots t^{(s_m)}_{k_ml_m} \cdot \eta,
	\]
	where $\xi$ is the highest weight vector of $M$.
	By the PBW theorem, we may choose this to be an ordered monomial, given a linear order on the generators.
	Choosing the order such that all elements of $\mathcal{C}$ are to the right, we see that the monomial cannot contain any elements of $\mathcal{C}$. 
	Now, we observe that the charge in $M$ of the left-hand side of the equation is zero. 
	Since the monomial contains no elements which reduce the charge in $M$, the only remaining possibility is that all elements in the monomial, as well as $\eta$ itself, have charge zero. 
\end{proof}

Finally, we end this section with a very important result for the Bethe ansatz, namely the commutation relation between $A$ and $B$ operators. 

\begin{lemma} \label{L:AB-relation}
	The block operators $A$ and $B$ satisfy
	\begin{multline} \label{AB-relation}
		(A^{I}_{I}(u))_a (B^{J}_{J+1}(v))_b = 
		(D^{IJ}(u-v))^{-1} (B^{J}_{J+1}(v))_b (A^{I}_{I}(u))_a  D^{I ,J+1}(u-v)
		\\
		+P^{I, J}_{J , I} (B^{I}_{I+1}(u))_b (A^{J}_{J}(v))_a P^{J, I+1}_{I+1, J} L_{I ,J+1}^{I+1, J}(u-v) \\
		-L_{I-1,J+1}^{I, J}(u-v) P_{J+1, I-1}^{I-1, J+1} (B^{I-1}_{I}(u))_b (A^{J+1}_{J+1}(v))_a P^{J+1, I}_{I ,J+1} + \mcC.
	\end{multline}
\end{lemma}

\begin{proof}
	Again we use the $UDL$ decomposition for the $R$ matrix Proposition~\ref{UDL-decomp}.
	Write $R(u) = U(u) D(u) L(u)$. 
	We have 
	\[
		L(u-v) T_a(u) T_b(v) L(u-v)^{-1} = D(u-v)^{-1} U(u-v)^{-1} T_b(v) T_a(u) U(u-v) D(u-v) .
	\]
	Here $L(u): V^I \ot V^J \to \bigoplus_{K\geq 0} V^{I+K} \ot V^{J-K}$.
	We will omit the spectral parameter in what follows.
	Taking the submatrix in $\Hom(V^I \ot V^{J+1}, V^I \ot V^J)$,
	\begin{align*}
		rhs &= (D^{I,J})^{-1} \sum_{i = 0}^{I-1} \sum_{j = 0}^{J-1} (U^{-1})^{I,J}_{I+j,J-j} (T_b)^{J-j}_{J+1+i} (T_a)^{I+j}_{I-i} U^{I-i,J+1+i}_{I ,J+1} D^{I , J+1}
		\\
		&= (D^{I,J})^{-1} (B^{J}_{J+1})_b (A^{I}_{I})_a  D^{I , J+1}  + \mcC
	\end{align*}
	The left-hand side is 
	\begin{align*}
		lhs &= \sum_{i = 0}^{I-1} \sum_{j = 0}^{J} L^{I , J}_{I-i, J+i} (T_a)^{I-i}_{I+j} (T_b)^{J+i}_{J+1-j} (L^{-1})_{I ,J+1}^{I+j,J+1-j}
		\\
		&=(A^{I}_{I})_a (B^{J}_{J+1})_b + L^{I , J}_{I-1, J+1} (B^{I-1}_{I})_a (A^{J+1}_{J+1})_b + (B^{I}_{I+1})_a (A^{J}_{J})_b (L^{-1})_{I,J+1}^{I+1, J} + \mcC
		\\
		&=(A^{I}_{I})_a (B^{J}_{J+1})_b + L^{I , J}_{I-1, J+1} (B^{I-1}_{I})_a (A^{J+1}_{J+1})_b - (B^{I}_{I+1})_a (A^{J}_{J})_b L_{I,J+1}^{I+1, J} + \mcC.
	\end{align*}
	Combining both sides, we insert permutation operators to obtain 
	\begin{multline*}
		(A^{I}_{I}(u))_a (B^{J}_{J+1}(v))_b = 
		(D^{I, J})^{-1} (B^{J}_{J+1}(v))_b (A^{I}_{I}(u))_a  D^{I ,J+1}
		\\
		+P^{I, J}_{J , I} B^{I,I+1}_b(u) (A^{J}_{J}(v))_a P^{J, I+1}_{I+1, J} L_{I ,J+1}^{I+1, J} \\
		-L_{I-1,J+1}^{I, J} P_{J+1, I-1}^{I-1, J+1} (B^{I-1}_{I}(u))_b (A^{J+1}_{J+1}(v))_a P^{J+1, I}_{I ,J+1} + \mcC.
	\end{multline*}

\end{proof}

The terms on the right-hand side of \eqref{AB-relation} are referred to as ``wanted'' (top row) and ``unwanted'' (second and third rows) terms, in relation to their contribution to the eigenvalue of the Bethe eigenvector. 
Specifically, the wanted terms are those with $A^{I}_{I}(u)$ for some $I$; the unwanted terms contain terms with $A^{I}_{I}(v)$.
It is generally shown that the Bethe equations are a sufficient condition for the vanishing of the unwanted terms. 

\section{The creation operator} \label{sec:onex}

In the previous section, we introduced the idea of the vacuum sector. 
In physical terms, we consider the vacuum sector as the space annihilated by all \emph{annihilation operators}, that is, below-diagonal blocks of $T(u)$. 
Continuing this analogy, we construct eigenvectors by acting on the vacuum sector by \emph{creation operators}. 

In general, creation operators will consist of nontrivial linear combinations of products of $B$ and $A$ operators. 
However, there is a lot to be learned just from the one-excitation case, involving just a single $B$ operator, which will be covered in this section. 

Indeed, we will start with the case of $\mathfrak{a}_1$, that is, the algebraic Bethe ansatz. 
We then bring these ideas to the more general case to make an ansatz for the creation operator for one excitation.


\subsection{The algebraic Bethe ansatz}

Consider the $\mathfrak{a}_1$ case: the algebraic Bethe ansatz for an $\mathfrak{sl}_2$-symmetric spin chain with auxiliary space $V$ of arbitrary spin. 
In this case, the decompositions of Section~\ref{s:monodromy} still hold, with the caveat that $\gs$ is simply the trivial Lie algebra, and so there is no Yangian subalgebra. 
Furthermore, the spaces $V^I$ are simply the weight spaces of $V$, and the ``block'' $UDL$ decomposition of Proposition~\ref{UDL-decomp} is now exactly the familiar triangular (Gauss) decomposition of the $R$-matrix. 
The exact matrix elements of the $D^{IJ}(u)$ were calculated in \cite{khoroshkinYangianDouble1996} using the universal $R$-matrix.
We find that
\[
	\frac{D^{I,J+1}(u)}{D^{IJ}(u)} = \frac{(u+J-N) (u+J+1)}{(u+J-I)(u+J-I+1)},
\]
where $N=\dim(V)+1$.
Additionally, the lower triangular matrix $L$ is given by
\[
	L(u) = 1 + \frac{1}{u+(h \ot 1)-(1\ot h)+1} e_{\alpha} \ot e_{-\alpha} + \dots 
\]
where the omitted terms correspond to larger changes in the charge. 
From this, we can rewrite the $AB$ relation from Lemma~\ref{L:AB-relation} as
\begin{multline*}
	a^I(u)b^{J}_{J+1}(v^J) = \frac{(u-v^J+J-N) (u-v^J+J+1) }{(u-v^J+J-I)(u-v^J+J-I+1)} b^{J}_{J+1}(v^J)a^I(u)
	\\
	+ \frac{b^{I}_{I+1}(u) a^{J}(v^J)}{u-v^J+J-I} 
	- \frac{b^{I-1}_{I}(u) a^{J+1}(v^J)}{u-v^J+J-I+1} 
	+ \mcC.
\end{multline*}
It is important to note that choosing $v^J=v+J$ gives a ``wanted term'' that is independent of $J$, equal to
\[
	\frac{(u-v-N) (u-v+1) }{(u-v-I)(u-v-I+1)} a^I(u).
\]
From this observation, we deduce that each $b^{J}_{J+1}(v+J)$ produces an equivalent excitation for each $J$, as the resulting Bethe vector has same transfer matrix eigenvalue. 
Taking the sum over $I$, the wanted terms combine to give a ``dressed transfer matrix'' for a single excitation, 
\[
	t'(u;v) := \sum_{I=0}^N \frac{(u-v-N) (u-v+1) }{(u-v-I)(u-v-I+1)}a^I(u),
\]
and the unwanted terms may be combined to give
\begin{equation*}
	t(u)  \beta(v) 
	=  \beta(v) t'(u;v)
	+ \sum_{I=0}^{N-1} \frac{b^{I}_{I+1}(u)}{u-v-I}  \left(  a^{J}(v+J) - a^{J+1}(v+J) \right) + \mcC,
\end{equation*}
which we then may identify as the residue of the dressed transfer matrix:
\[
	\Res_{u \to v+I} t'(u;v) := (I+1)(I-N)\left(a^I(v+I)- a^{I+1}(v+I) \right).
\]
The resulting expression is
\begin{equation}
	t(u) \beta(v) 
	=  \beta(v) t'(u;v)
	\\
	+ \frac{1}{(J+1)(J-N)} \sum_{I=0}^{N-1} \frac{b^{I}_{I+1}(u)}{u-v-I} \Res_{w \to v+J} t'(w;v) 
	+ \mcC.
\end{equation}

For multiple excitations, the $b$ operators must be combined nontrivially.

\subsection{The auxiliary site}

For the general $\g$ case, the construction of the creation operator is more involved, as the block matrices $B^{I}_{J}(v)$ are matrix operators.
Indeed, since this operator acts on a nontrivial auxiliary space, it suggests that this auxiliary space must be appended to $M^0$ in some way. 

We propose that, in fact, the sites that are added to the spin chain should all be in the \emph{same} $\gs$-representation, regardless of $V$,
\begin{equation} \label{auxsite}
	V^{aux} := \g^{(1)} \cong \mathrm{span}_{\C}\left\{ x^-_{\alpha} \in \g \middle| \alpha >0, (\omega^{\vee}_p,\alpha) = -1 \right\},
\end{equation}
where the action is given by the restriction of the adjoint action of $\g$. 
We will refer to this space as the \emph{auxiliary site}, and define the \emph{nested vacuum sector} to be the space $(V^{aux})^{\ot m} \ot M^0$.

\begin{lemma}\label{l:properties-of-auxsite}
	The auxiliary site $V^{aux}$, defined in \eqref{auxsite}, satisfies the following:
	\begin{enumerate}
		\item $V^{aux}$ is irreducible as a representation of $\gs$, and has highest weight $-\pi(\alpha_p)$, where $\pi$ is the projection from the weight lattice of $\g$ to the weight lattice of $\gs$. 
		\item An action of $Y_{\hbar}(\gs)$ may be defined on $V^{aux}$ such that $V^{aux}$ is irreducible as a $Y_{\hbar}(\gs)$ representation.
		\item There is a non-trivial intertwiner of $\gs$-representations $V^{I+J} \ot (V^I)^*$ and $(V^{aux})^{\ot J}$.
	\end{enumerate}
\end{lemma}

\begin{proof}
	The first part follows from the fact that $\g$ is simple, and that $-\alpha_p$ is clearly the highest root of all negative roots which contain $-\alpha_p$. 
	
	For the second part, we recall the representation $\g \oplus \C$ of $Y_{\hbar}(\g)$ from Theorem~8 of \cite{drinfeldHopfAlgebrasQuantum1985}, see also Section~5 of \cite{pressleyFundamentalRepresentationsYangians1991} for a more detailed exposition.
	Restricting to $Y_{\hbar}(\gs)$, gives the result. 
	
	For the third part, we will show equivalently that there is a $\gs$-homomorphism $(V^{aux})^{\ot J} \ot V^I \to V^{I+J}$.
	Recall that $V^{aux}$ has a basis given by weight vectors $e_{-\alpha}$. 
	For $J=1$, we clearly have a map
	\[
		\psi(e_{-\alpha} \ot \zeta) = e_{-\alpha} \cdot \zeta,
	\]
	for $\zeta \in V^{I}$, regarding $V^I$ as a subspace of $V$. 
	This is a homomorphism because
	\[
		x \cdot \psi(e_{-\alpha} \ot \zeta) = x e_{-\alpha} \cdot \zeta = ([x, e_{-\alpha}] + e_{-\alpha} x) \cdot \zeta 
		= x \cdot e_{-\alpha} + e_{-\alpha} x \cdot \zeta = \psi(\Delta(x) \cdot  e_{-\alpha} \ot  \zeta ).
	\]
	This then extends to the $J$-fold tensor product of $(V^{aux})^{\ot J}$ by composing the individual homomorphisms.
\end{proof}

\begin{remark}
	The Lie algebra weight of this space agrees with that of the auxiliary space whose Drinfel'd polynomials appear in the Bethe equations, as in \eqref{nested-drin-polys}.
\end{remark}

The non-trivial intertwiners of Lemma~\ref{l:properties-of-auxsite} bestow a natural action on the nested vacuum sector to the $B$ operators. 
This puts us in a position to define the one-particle state.
Let us introduce the $\covec$ operator, which is the natural isomorphism 
\[
	\covec: \Hom(V,W) \to V^* \ot W.
\]
For $M \in \End(W), B \in \Hom(V,W), N \in \End(V)$
\[
	\covec(M B N) = \covec(B) (N \ot M^t).
\]
Let us write $\beta_{J,\overline{I}}(u) := \covec{B^{I}_{J}}(u)$. 
Then the \emph{one-particle creation operator} is 
\[
	\beta_{J+1, \overline{J}}(u) Z^{J+1} \Gamma^{J+1, \overline{J}}_{aux} \in \Hom(V^{aux},\C)  \otimes \X,
\]
where $\Gamma^{J+1, \overline{J}}_{aux}$ is an intertwiner of $\gs$ representations $V^{aux} \hookrightarrow V^{J+1} \ot (V^{J})^*$.
The validity of this expression for the creation operator must be borne out by the action of the transfer matrix, which we will explore in Section~\ref{ssec:wanted-terms}.

Let us take a slight detour here and discuss the case where $V\cong V(\omega_p)$. 
In this case, we can be more explicit. 
Indeed, $V^0$ must be trivial, and so $V^{aux} \cong V^1$.
The Drinfel'd polynomials for $V^1$ are then obtained from Lemma 4.4 of \cite{tanBraidGroupActions2015}, giving 
\begin{equation}
	P^{V^{aux}}_i(u) = 
	\begin{cases}
		1 & a_{ip}=0 \\
	 	u-\frac{d_i \hbar}{2} & a_{ip}=-1; \\
		(u-\hbar)(u) & a_{ip}=-2 \\ 
		(u-\tfrac{3\hbar}2)(u-\tfrac{\hbar}2 )(u+\tfrac{\hbar}2 ) & a_{ip}=-3.
	\end{cases}
\end{equation}
In this case, $B^0_1(u) \in \Hom(V^{aux}\ot M, M)$ defines an appropriate creation operator. 

In general, however, we will find that we require the following property to proceed with the nested Bethe ansatz.
This property was noted by Reshetikhin in \cite{reshetikhinAlgebraicBetheAnsatz1988} for the case $N=1$.

\begin{conjecture} \label{c:projector}
	Write $\Pi^{J+1,\overline{J}}$ for the projector to $V^{aux}$ in $V^{J+1} \ot (V^{J})^*$. Then for some $c\in \C$ and $k \in \mathbb{Z}$,
	\[
		\left[(D^{J+1,J}(u))^{-1}\right]^{t_{J}} =  c\, \Pi^{J+1,\overline{J}} \, u^{k} + O(u^{k+1}) \quad \text{as} \quad u \to 0.
	\]
\end{conjecture}

\begin{remark}
	For many $R$-matrices, the leading coefficient of $R(u)$ as $u \to 0$ is proportional to the permutation matrix $P$; these are often referred to as regular $R$-matrices.
	It is clear, however, that $P$ does not admit a block Gauss decomposition. 
	Indeed, while the restriction $P^{II}_{II}$ simply gives the permutation on $V^I \ot V^I$, the restriction $P^{IJ}_{IJ}$ for $I\neq J$ is necessarily equal to zero.
	This implies that $D^{IJ}(u)$ has a singularity at $u=0$ if $R(u)$ is a regular $R$-matrix. 
\end{remark}

Let us provide a few examples where the conjecture holds.
In the case where $\gs$ is classical, and all representations involved are fundamental or trivial, the result can be checked by a straightforward application of the results of \cite{chariYangiansIntegrableQuantum1996}.

\begin{example}
	Consider the nesting $(\g, \gs) = (\mf{d}_6,\mf{a}_5)$ with $V=M^{\mf{d}_6}(\omega_1)$.
	From the Baxter polynomials, we confirm that 
	\[
		M^{\mf{d}_6}(\omega_1)_0 \cong M^{\mf{a}_5}(\omega_1)_0 \oplus M^{\mf{a}_5}(\omega_5)_{2\hbar}.
	\]
	As a representation of the Yangian, we have $(M^{\mf{a}_5}(\omega_1)_0)^* \cong M^{\mf{a}_5}(\omega_5)_{3\hbar}$.
	The auxiliary site is $V^{aux} \cong M^{\mf{a}_5}(\omega_4)_u$, for some $u$.
	From Theorem~6.1 of \cite{chariYangiansRmatrices1990} in this case, and states that
	\[
		\Hom_{Y_{\hbar}(\gs)}\left( M^{\mf{a}_5}(\omega_5)_{2\hbar}  \otimes M^{\mf{a}_5}(\omega_5)_{3\hbar}, M^{\mf{a}_5}(\omega_4)_{u}\right) \neq 0
	\]
	if and only if $u=\tfrac{5\hbar}{2}$.
\end{example}

\begin{example}
	Consider the nesting $(\g, \gs) = (\mf{e}_6,\mf{d}_5)$ with $V=M^{\mf{e}_6}(\omega_1)$.
	From Table~\ref{t:v-decomp}, we confirm that 
	\[
		M^{\mf{e}_6}(\omega_1)_0 \cong M^{\mf{d}_5}(\omega_1)_0 \oplus M^{\mf{d}_5}(\omega_5)_{3\hbar/2} \oplus 1.
	\]
	In the Yangian, we have $(M^{\mf{d}_5}(\omega_1)_0)^* \cong M^{\mf{d}_5}(\omega_1)_{4\hbar}$.
	The auxiliary site is $V^{aux} \cong M^{\mf{d}_5}(\omega_5)_u$, for some $u$.
	Using Theorem~7.1 of \cite{chariYangiansRmatrices1990},
	\[
		\Hom_{Y_{\hbar}(\gs)}\left(M^{\mf{d}_5}(\omega_5)_{3\hbar/2} \otimes M^{\mf{d}_5}(\omega_1)_{4\hbar}, M^{\mf{d}_5}(\omega_5)_{u}\right) \neq 0
	\]
	if and only if $u=\tfrac{5\hbar}{2}$.
\end{example}


From Proposition~\ref{p:nested-ybe} and Conjecture~\ref{c:projector}, we have the fusion relation in $V^{I}\ot V^{J+1} \ot (V^J)^* $,
\[
	D^{I,J+1}(u)[D^{IJ}(u)^{-1}]^{t_J}\Pi^{J+1,\overline{J}} = \Pi^{J+1,\overline{J}} [D^{IJ}(u)^{-1}]^{t_J} D^{I,J+1}(u),
\]
leading to the fused $R$-matrices
\begin{equation} \label{fused-R}
	\Pi^{J+1,\overline{J}} D^{I,J+1}(u)[D^{IJ}(u)^{-1}]^{t_J}\Pi^{J+1,\overline{J}}.
\end{equation}

Fusion as conjectured above allows us to make sense of the idea that each one-particle creation operator should be equivalent with the appropriate shifts, which was observed in the rank one case. 
Therefore, let us pick a particular $J$ and write
\[
	R^{I,aux}(u) := \Pi^{J+1,\overline{J}} D^{I,J+1}(u)[D^{IJ}(u)^{-1}]^{t_J}\Pi^{J+1,\overline{J}}.
\]
By uniqueness, the fused $R$-matrix for $K\neq J$ is related to this one by a parameter shift and rational function;
\[
	\Pi^{K+1,\overline{K}} D^{I,K+1}(u)[D^{IK}(u)^{-1}]^{t_K}\Pi^{K+1,\overline{K}} = f^{K}(u)R^{I,aux}(u+w^K).
\]

\subsection{The wanted terms and the nested Bethe ansatz} \label{ssec:wanted-terms}

From Lemma~\ref{L:AB-relation}, applying the $\covec$ operator gives the equation
\begin{multline*}
	(A^{I}_{I}(u))_a (\beta_{J+1,\overline{J}}(v))_{b \bt} = 
	(\beta_{J+1,\overline{J}}(v))_{b \bt} ((D^{I , J}_{a \bt}(u-v))^{-1})^{t_{\bt}} (A^{I}_{I}(u))_a  D^{I ,J+1}_{ab}(u-v)
	\\
	+ (\beta_{I+1, \overline{I}}(u))_{b \bt} (P_{I , J}^{J , I})^{t_{\bt}}_{a \bt} (A^{J}_{J}(v))_a P^{J, I+1}_{I+1, J} L_{I ,J+1}^{I+1, J} (u-v)
	\\
	-(\beta_{I,\overline{I-1}}(u))_{b \bt} (L_{I-1,J+1}^{I, J}(u-v))^{t_{\bt}} (P_{J+1, I-1}^{I-1, J+1})^{t_{\bt}}_{a \bt} (A^{J+1}_{J+1}(v))_a  P^{J+1, I}_{I ,J+1} + \mcC.
\end{multline*}
We want to transform this expression to give a commutation relation for the transfer matrix $t(u)$.
Let us focus on the ``wanted terms'' only, in which case we have 
\[
	(A^{I}_{I}(u))_a (\beta_{J+1,\overline{J}}(v))_{b \bt} = 
	(\beta_{J+1,\overline{J}}(v))_{b \bt} ((D^{I,J}_{a \bt}(u-v))^{-1})^{t_{\bt}} (A^{I}_{I}(u))_a  D^{I,J+1}_{ab}(u-v)
	+  \text{u.t.},
\]
where ``u.t.'' stands for unwanted terms. 
Let $Z^I$ denote the restriction of $Z$ to the space $V^I$.
Multiply on the right by $Z^I_a Z^{J+1}_b$, to give
\begin{align*}
	&(A^{I}_{I}(u))_a Z^{I}_a (\beta_{J+1,\overline{J}}(v))_{b \bt} Z^{J+1}_b 
	\\
	&\qquad = 
	(\beta_{J+1,\overline{J}}(v))_{b \bt} ((D^{I,J}_{a \bt}(u-v))^{-1})^{t_{\bt}} (A^{I}_{I}(u))_a  D^{I,J+1}_{ab}(u-v) Z^I_a Z^{J+1}_b
	+  \text{u.t.}
	\\
	&\qquad =	(\beta_{J+1,\overline{J}}(v))_{b \bt} Z^{J+1}_b ((D^{I,J}_{a \bt}(u-v))^{-1})^{t_{\bt}} (A^{I}_{I}(u))_a Z^{I}_a  D^{I,J+1}_{ab}(u-v)
	+  \text{u.t.},
\end{align*}
where we have used the fact that $D^{I,J+1}_{ab}(u-v)$ is a $\gs$-intertwiner.
Taking the trace and applying cyclicity gives 
\begin{multline*}
	\tr_a \left(Z^I_a (A^{I}_{I}(u))_a \right) (\beta_{J+1,\overline{J}}(v))_{b \bt} Z^{J+1}_b \\= (\beta_{J+1,\overline{J}}(v))_{b \bt} Z^{J+1}_b \tr_a \left( Z^I_a D^{I,J+1}_{ab}(u-v) ((D^{I,J}_{a \bt}(u-v))^{-1})^{t_{\bt}} (A^{I}_{I}(u))_a \right)
	+  \text{u.t.}
\end{multline*}
Now, acting from the right with the projector, and applying relation \eqref{fused-R}, we have 
\begin{align*}
	&\tr_a \left(Z^I_a (A^{I}_{I}(u))_a \right) \beta_{B}(v) 
	\\
	&\qquad = (\beta_{J+1,\overline{J}}(v))_{b \bt} Z^{J+1}_b \tr_a \left( Z^I_a D^{I,J+1}_{ab}(u-v) ((D^{I,J}_{a \bt}(u-v))^{-1})^{t_{\bt}} (A^{I}_{I}(u))_a \right) \Pi_{b \bt}+  \text{u.t.}
	\\
	&\qquad = \beta_{B}(v)  \tr_a \left(Z^I_a R^{I,aux}_{a B}(u-v) (A^{I}_{I}(u))_a \right) +  \text{u.t.}
\end{align*}
Here $\beta_B(v) = \beta_{b \bt}(v) \Pi^{J+1, \overline{J}}_{b \bt}$. 
Taking the sum over each $I$ gives the relation for the transfer matrix
\begin{equation}
	t(u) \beta_{B}(v) = \beta_{B}(v)  \sum_{I} \tr_a \left(Z^I_a R^{I,aux}_{a B}(u-v) (A^{I}_{I}(u))_a  \right) +  \text{u.t.}
\end{equation}

This relation is the foundational step of the nested Bethe ansatz. 
It connects the original transfer matrix diagonalisation problem, defined in terms of the data $(\g, V, M, Z)$, to a nested problem for the transfer matrices
$\tr_a \left(Z_a R^{I,aux}_{a B}(u-v) (A^{I}_{I}(u))_a  \right)$, which are defined up to a scalar function by the data $(\gs, V^I,V^{aux} \ot M^0, Z^I)$.
As a consequence of Corollary~\ref{cor:commuting-transfer}, these transfer matrices all mutually commute, and from Proposition~\ref{P:DAA}, we may regard the matrices $R^{I,aux}_{a B}(u-v) (A^{I}_{I}(u))_a$ as monodromy matrices for each $I$, from which the creation operators for the nested problem are derived, and so on. 

Now suppose $\Psi\in V^{aux} \ot M^0$ is a joint eigenvector of $\tr_a \left(Z_a R^{I,aux}_{a B}(u-v) (A^{I}_{I}(u))_a  \right)$ for each $I$, with eigenvalues $\Lambda^I(u;v)$. Suppose further that $v$ is chosen such that the unwanted terms annihilate $\Psi$.  
Then
\[
	t(u) \beta(v)\cdot \Psi := \sum_I \Lambda^I(u;v) \beta(v)\cdot \Psi,
\]
solving the original problem.

\subsection{The nested Bethe ansatz for classical Lie algebras}

Consider now the case where $N=1$; that is, the space $V$ splits into just two charge subspaces.
This case corresponds to those studied in the classic papers 
\cite{kulishDiagonalisationGLInvariant1983, reshetikhinAlgebraicBetheAnsatz1988, reshetikhinIntegrableModelsQuantum1985,devegaExactBetheAnsatz1987}.
Indeed, for each of the classical Lie algebras $\mathfrak{a}_r$, $\mathfrak{b}_r$, $\mathfrak{c}_r$, and $\mathfrak{d}_r$, there is a representation $V$ for which this condition holds. 
These are listed below, up to duality.

\begin{table}[ht] \label{t:N1case}
	\def\arraystretch{1.3}
	\centering

\begin{tabular}{CCCCC}
	\g \downarrow \gs &  V & V^0 & V^1 & V^{aux}\\
	\hline
	\mathfrak{a}_r \downarrow \mf{a}_{r-1} & M^{\g}(\omega_1) & 1 & M^{\gs}(\omega_1)_{\hbar/2} & M^{\gs}(\omega_1)_{\hbar/2} \\[2pt]
	\mathfrak{b}_r \downarrow \mf{b}_{r-1} & M^{\g}(\omega_r) & M^{\gs}(\omega_{r-1}) & M^{\gs}(\omega_{r-1})_{2\hbar} & M^{\gs}(\omega_1)_{(r-3)\hbar}\\[2pt]
	\mathfrak{c}_r \downarrow \mf{a}_{r-1} & M^{\g}(\omega_1) & M^{\gs}(\omega_{1}) & M^{\gs}(\omega_{r-1})_{(r+2)\hbar/2} & M^{\gs}(\omega_{r-1})_{(r+3)\hbar/2}\\[2pt]
	\mathfrak{d}_r \downarrow \mf{d}_{r-1} & M^{\g}(\omega_r) & M^{\gs}(\omega_{r-1}) & M^{\gs}(\omega_{r-2})_{\hbar} & M(\omega_{1})_{(r-1)\hbar/2}\\[2pt]
	\mathfrak{d}_r \downarrow \mf{a}_{r-1} & M^{\g}(\omega_1) & M^{\gs}(\omega_{1}) & M^{\gs}(\omega_{r-1})_{(r-2)\hbar/2} & M(\omega_{r-2})_{(r-1)\hbar/2}
\end{tabular}

\end{table}

In this case, there are a number of simplifications in the above theory. 
\begin{itemize}
	\item When $\Jmax=1$, all but one of the $D$ matrices are simply submatrices of $R(u)$, with the final one being given by 
	\begin{equation} \label{D12}
		D^{01}(u) = [P R^{10}_{10}(-u) P]^{-1}.
	\end{equation}
	\item There is only one choice for $B^{J}_{J+1}(v) = B^{0}_{1}(v) =: B(v)$ and, moreover, there are no $B$-operators for multiple excitations. This means that the creation operator for multiple excitations is given by a unique product of $B$-operators, as in the standard algebraic Bethe ansatz.
	\item Since the $B$-operators map from the highest charge subspace to the lowest charge subspace, we immediately have, from conservation of charge, 
	\begin{equation} \label{RBB}
		D^{00}_{b_1b_2}(v_1-v_2) B_{b_1}(v_1) B_{b_2}(v_2) = B_{b_2}(v_2) B_{b_1}(v_1) D^{11}_{b_1b_2}(v_1-v_2).
	\end{equation}
	\item The AB relation becomes true in general, rather than just on the vacuum sector. This can be seen by the fact that there is no product of two matrices in $\mcC$ with the appropriate charge. Further, each AB relation contains only one unwanted term, and are given explicitly by 
	\begin{multline*}
		(A^{0}_{0}(u))_a \beta_{b \bt}(v) = 
		\beta_{b \bt}(v) ((D^{00}_{a \bt}(u-v))^{-1})^{t_{\bt}} (A^{0}_{0}(u))_a  D^{01}_{ab}(u-v)
		\\
		+ \beta_{b \bt}(u) (P_{00}^{00})^{t_{\bt}}_{a \bt} (A^{0}_{0}(v))_a (P^{01}_{10})_{ab} (L_{01}^{10} (u-v))_{ab},
	\end{multline*}
	and 
	\begin{multline*}
		(A^{1}_{1}(u))_a \beta_{b \bt}(v) = 
		\beta_{b \bt}(v) ((D^{10}_{a \bt}(u-v))^{-1})^{t_{\bt}} (A^{1}_{1}(u))_a  D^{11}_{ab}(u-v)
		\\
		-\beta_{b \bt}(u) (L_{01}^{10}(u-v))^{t_{\bt}}_{a\bt} (P_{10}^{01})^{t_{\bt}}_{a \bt} (A^{1}_{1}(v))_a  (P^{11}_{11})_{ab}.
	\end{multline*}
\end{itemize}

These simplifications leave us with only one possibility for the structure of the Bethe vector. 
However, due to the fusion, the exact way to combine the $B$-operators is slightly nontrivial.
We show below a setup which gives the desired properties of the multiparticle creation operator.



\begin{proposition} The product of operators 
	\begin{multline}
		\beta_{B_1 \dots B_m}(\vec{v}) = 
		\beta_{\bt_1 b_1}(v_1) \cdots \beta_{\bt_m b_m}(v_m)
		Z_{b_1} \cdots Z_{b_m} \times 
		\\
		\times \prod_{k=1}^m  \left( (D^{11}_{b_k b_{k+1}})^{-1} (D^{01}_{\bt_k b_{k+1}})^{t_{\bt_k}} \cdots (D^{11}_{b_{m-1} b_{m}})^{-1} (D^{01}_{\bt_{m-1} b_{m}})^{t_{\bt_{m-1}}}  \right)
		\prod_{k=1}^m \Pi_{b_k \bt_k} 
	\end{multline}
	satisfies
	\begin{gather}
			\beta_{\bm B}(\vec{v}) = \beta_{\bm B}(v_1, \dots v_{k+1}, v_k, \dots, v_m) \check{R}_{B_k B_{k+1}}(v_k-v_{k+1})
			\\
			\tr_a Z^I_a (A^{I}_{I}(u))_a \beta_{\bm B}(\vec{v}) = \beta_{\bm B}(\vec{v}) \tr_a \left( Z^I_a R^{I,aux}_{a B_1}(u-v_1) \cdots R^{I,aux}_{a B_m}(u-v_m) (A^{I}_{I}(u))_a \right) + \text{u.t.},
	\end{gather}
	where $\check{R}(u) = P R(u)$. 
\end{proposition}

\begin{proof}
	We will prove the case for $m=2$ only; the rest will follow similarly. 
	From \eqref{RBB}, we obtain 
	\[
		\beta_{\bt_1 b_1}(v_1) \beta_{\bt_2 b_2}(v_2)
		=
		\beta_{\bt_2 b_2}(v_2) \beta_{\bt_1 b_1}(v_1)
		 D^{00}_{\bt_1 \bt_2}(v_1-v_2)^{-1}
		 D^{11}_{b_1 b_2}(v_1-v_2)
	\]
	Therefore, multiplying on the left and right by the relevant $R$-matrices, 
	\begin{equation} \label{bbstart}
		\beta_{B_1B_2}(v_1,v_2)=\beta_{\bt_2 b_2}(v_2) \beta_{\bt_1 b_1}(v_1) \left(D^{00}_{\bt_1 \bt_2}\right)^{-1}
		D^{11}_{b_1 b_2} [D^{11}_{b_1 b_2}]^{-1} [D^{01}_{\bt_1 b_2}]^t  \Pi_{b_1 \bt_1} \Pi_{b_2 \bt_2}
	\end{equation}
	We have omitted the spectral parameter dependence. 
	It is a matter of rearranging this product of $R$-matrices using the Yang-Baxter equation. 
	The following sequence of Yang-Baxter moves gives the appropriate result; we have underlined uses of the Yang-Baxter equation, or simple commutation of matrices acting on different spaces:
	\begin{align*}
		\bm{R} :&= [D^{00}_{\bt_1 \bt_2}]^{-1} D^{11}_{b_1 b_2} [D^{11}_{b_1 b_2}]^{-1} [D^{01}_{\bt_1 b_2}]^t  \Pi_{b_1 \bt_1} \Pi_{b_2 \bt_2}
		\\
		&= [D^{00}_{\bt_1 \bt_2}]^{-1} [D^{01}_{\bt_1 b_2}]^t  \Pi_{b_1 \bt_1} \Pi_{b_2 \bt_2}
		\\
		&= [D^{00}_{\bt_1 \bt_2}]^{-1} [D^{01}_{\bt_1 b_2}]^t  \underline{\Pi_{b_2 \bt_2} \Pi_{b_1 \bt_1}} 
		\\
		&= \left( D^{01}_{\bt_1 b_2} [[D^{00}_{\bt_1 \bt_2}]^{-1}]^t \Pi_{b_2 \bt_2} \right)^{t_{\bt_1}}   \Pi_{b_1 \bt_1}
		\\
		&= [R^{1,aux}_{\bt_1 B_2}]^{t_{\bt_1}} \Pi_{b_1 \bt_1} 
		\\
		&=  [D^{11}_{b_2 b_1}]^{-1} [D^{01}_{\bt_2 b_1}]^t  [[D^{01}_{\bt_2 b_1}]^t]^{-1} D^{11}_{b_2 b_1} \Pi_{b_2 \bt_2}  [R^{1,aux}_{\bt_1 B_2}]^{t_{\bt_1}} \Pi_{b_1 \bt_1} 
		\\
		&=  [D^{11}_{b_2 b_1}]^{-1} [D^{01}_{\bt_2 b_1}]^t  \underline{[[D^{01}_{\bt_2 b_1}]^t]^{-1} D^{11}_{b_2 b_1} \Pi_{b_2 \bt_2} } [R^{1,aux}_{\bt_1 B_2}]^{t_{\bt_1}} \Pi_{b_1 \bt_1} 
		\\
		&=  [D^{11}_{b_2 b_1}]^{-1} [D^{01}_{\bt_2 b_1}]^t  \left(R^{2,aux}_{b_1 B_2}(-u) \right)^{-1} [R^{1,aux}_{\bt_1 B_2}]^{t_{\bt_1}} \Pi_{b_1 \bt_1} 
		\\
		&=  [D^{11}_{b_2 b_1}]^{-1} [D^{01}_{\bt_2 b_1}]^t \Pi_{b_1 \bt_1} \Pi_{b_2 \bt_2}  R^{aux, aux}_{B_2 B_1}
	\end{align*}
	Hence, returning to the \eqref{bbstart},
	\[
		\beta_{B_1B_2}(v_1,v_2)=\beta_{\bt_2 b_2}(v_2) \beta_{\bt_1 b_1}(v_1)  [D^{22}_{b_2 b_1}]^{-1} [D^{12}_{\bt_2 b_1}]^t \Pi_{b_1 \bt_1} \Pi_{b_2 \bt_2} R^{aux, aux}_{B_2 B_1}(v_2-v_1).
	\]
	To complete the proof, we insert permutation operators $P_{b_1 b_2} P_{\bt_1  \bt_2} = P_{B_1 B_2}$.
	This gives
	\begin{align*}
		\beta_{B_1B_2}(v_1,v_2)&=
		\beta_{\bt_2 b_2}(v_2) \beta_{\bt_1 b_1}(v_1)  [D^{22}_{b_2 b_1}(v_2-v_1)]^{-1} [D^{12}_{\bt_2 b_1}(v_2-v_1)]^t \times
		\\
			&\hspace{5cm} \times \Pi_{b_1 \bt_1} \Pi_{b_2 \bt_2} P_{b_1 b_2} P_{\bt_1  \bt_2} P_{B_1 B_2} R^{aux, aux}_{B_2 B_1}(v_2-v_1)
		\\
		&=\beta_{\bt_2 b_2}(v_2) \beta_{\bt_1 b_1}(v_1) P_{b_1 b_2} P_{\bt_1  \bt_2} [D^{22}_{b_1 b_2}(v_2-v_1)]^{-1} [D^{12}_{\bt_1 b_2}(v_2-v_1)]^t \times 
		\\
			&\hspace{7.5cm} \times \Pi_{b_1 \bt_1} \Pi_{b_2 \bt_2}  \check{R}^{aux, aux}_{B_2 B_1}(v_2-v_1)
		\\
		&=\beta_{B_1B_2}(v_2,v_1) \check{R}^{aux, aux}_{B_2 B_1}(v_2-v_1),
	\end{align*}
	where we have used $\beta_{\bt_2 b_2}(v_2) \beta_{\bt_1 b_1}(v_1) P_{b_1 b_2} = \beta_{\bt_1 b_1}(v_2) \beta_{\bt_2 b_2}(v_1) P_{\bt_1 \bt_2}$.

	For the AB relation, we obtain
	\begin{align*}
		&\tr_a (A^{I}_{I}(u))_a	 \beta_{B_1B_2}(v_1,v_2) 
		\\
		\qquad &=\beta_{b_1 \bt_1}(v_1) \beta_{b_2 \bt_2}(v_2) \tr_a  D^{I1}_{a b_2} D^{I1}_{a b_1} (D^{I0}_{a \bt_1})^{-t} (D^{I0}_{a \bt_2})^{-t} (A^{I}_{I}(u))_a
		[D^{11}_{b_1 b_2}]^{-1} [D^{01}_{\bt_1 b_2}]^t 
		\\
		\qquad &=\beta_{b_1 \bt_1}(v_1) \beta_{b_2 \bt_2}(v_2) \tr_a \underline{[D^{11}_{b_1 b_2}]^{-1} D^{I1}_{a b_1} D^{I1}_{a b_2}}   (D^{I0}_{a \bt_1})^{-t} (D^{I0}_{a \bt_2})^{-t} (A^{I}_{I}(u))_a [D^{01}_{\bt_1 b_2}]^t
		\\
		\qquad &=\beta_{b_1 \bt_1}(v_1) \beta_{b_2 \bt_2}(v_2)  [D^{11}_{b_1 b_2}]^{-1} \tr_a D^{I1}_{a b_1} \underline{[D^{01}_{\bt_1 b_2}]^t (D^{I0}_{a \bt_1})^{-t} D^{I1}_{a b_2} } (D^{I0}_{a \bt_2})^{-t} (A^{I}_{I}(u))_a 
		\\
		\qquad &=\beta_{b_1 \bt_1}(v_1) \beta_{b_2 \bt_2}(v_2)  [D^{11}_{b_1 b_2}]^{-1} [D^{01}_{\bt_1 b_2}]^t \tr_a D^{I1}_{a b_1} (D^{I0}_{a \bt_1})^{-t} D^{I1}_{a b_2} (D^{I0}_{a \bt_2})^{-t} (A^{I}_{I}(u))_a 
		\\
		\qquad &= \beta_{B_1B_2}(v_1,v_2) \tr_a \left( R^{I,aux}_{aB_1}(u-v_1) R^{I,aux}_{aB_2}(u-v_2) (A^{I}_{I}(u))_a \right)
	\end{align*}.
\end{proof}



\section{Conclusion and outlook}

In this article, we have introduced a new way of thinking about the nested Bethe ansatz using the $U(1)$ charge and the block Gauss decomposition. 
This theory combines a number of known results into a single theory, while opening the path to generalisations. 
Of course, this theory is far from complete, and we will list the next steps that must be taken to round out the theory below. 

First, we require a proof of Conjecture~\ref{c:projector}. 
It seems that this result may hold along the same lines as Dorey's rule \cite{chariYangiansIntegrableQuantum1996,doreyRootSystemsPurely1991}.
The next piece is that of the unwanted terms, which have been ignored in this article. 
It is often shown that the unwanted terms may be written as a residue of the wanted term; in the case of multiple creation operators, it seems that the residue must be taken at multiple points. 
Preliminary testing indicates that these points are the poles of $Q_{p,V}(u)$, the final Baxter polynomial.
This indicates that a proof of this could follow from the results of \cite{gautamPolesFinitedimensionalRepresentations2023}.
Finally, we note that our expression for the AB relation holds only for one excitation.
The next big step is to create an expression for the creation operator for multiple excitations when $N>1$. 
This type of calculation is familiar from \cite{martinsAlgebraicBetheAnsatz1997}. 
It may be possible that a universal version of this result follows from generalisation of the same triangularity argument. 

More generally, one can consider the standard vectors for generalisation in the field of integrability. 
The author is confident that these results may be generalised to the quantum affine algebra case without significant effort as the representation theory of the two is known to be equivalent \cite{nakajimaQuiverVarietiesFinite2001,varagnoloQuiverVarietiesYangians2000}.
The supersymmetric cases present a somewhat larger challenge.
Finally, another direction is the systems corresponding to open spin chains, with twisted Yangian \cite{olshanskiiTwistedYangiansInfinitedimensional1992} or reflection algebra symmetry \cite{molevRepresentationsReflectionAlgebras2002}. 
In this case, the nesting of Lie algebras is replaced by that of symmetric pairs, and it would be interesting to investigate the effect of removing a simple root. 

\section{Acknowledgements}

Much of this research was conducted while the author was financially supported by JSPS International Research Fellowship.
The author would like to thank his host Prof. Tetsuo Deguchi for interesting discussions, and for providing a welcoming working environment.
The author is currently financially supported by a Tokyo University of Science research fellowship. 
In particular, he thanks Prof. Kazumitsu Sakai for this opportunity.

\bibliography{mylibrary}

\end{document}